\documentclass[11pt]{article}

\usepackage[final]{acl}

\usepackage{times}
\usepackage{latexsym}

\usepackage[T1]{fontenc}
\usepackage[utf8]{inputenc}

\usepackage{microtype}

\usepackage{inconsolata}
\usepackage{algorithm}
\usepackage{algorithmic}
\usepackage{subfigure}
\usepackage{graphicx}
\usepackage{makecell}
\usepackage{amsmath}
\usepackage{amssymb}
\usepackage{enumitem}
\usepackage{amsthm}
\usepackage{booktabs}
\theoremstyle{plain}
\newtheorem{theorem}{Theorem}
\newtheorem{proposition}[theorem]{Proposition}

\theoremstyle{definition}
\newtheorem{definition}[theorem]{Definition}

\theoremstyle{remark}
\newtheorem{remark}[theorem]{Remark}
\usepackage{multirow}
\usepackage{xr}
\usepackage[most]{tcolorbox}

\title{SharedRequest: Privacy-Preserving Model‑Agnostic Inference for Large Language Models}

\author{
 \textbf{Peihua Mai\textsuperscript{1,2,3}},
 \textbf{Xuanrong Gao\textsuperscript{3}},
 \textbf{Youlong Ding\textsuperscript{4}},
 \textbf{Xianglong Du\textsuperscript{5,6}},
\\
 \textbf{Wei Liu\textsuperscript{5,6}},
 \textbf{Yan Pang\textsuperscript{1,2,3}\thanks{Corresponding author: jamespang@nus.edu.sg}}
\\
 \textsuperscript{1}National University of Singapore (Chongqing) Research Institute,\\
 \textsuperscript{2}Chongqing Key Laboratory of Trusted Perception and Interaction Technology\\ for Intelligent and Connected Vehicles, Chongqing, China,\\
 \textsuperscript{3}National University of Singapore,
 \textsuperscript{4}Hebrew University of Jerusalem,\\
 \textsuperscript{5}State Key Laboratory of Intelligent Vehicle Safety Technology, Chongqing, China\\
 \textsuperscript{6}CHONGQING CHANGAN AUTOMOBILE Co., Ltd, Chongqing, China
}

\begin{document}
\maketitle
\begin{abstract}
With the widespread deployment of public large language models (LLMs) such as ChatGPT, protecting user prompt privacy has become an increasingly critical issue. Existing privacy-preserving inference methods sacrifice either utility or efficiency, and often require model-specific modifications that limit their compatibility. In this paper, we propose SharedRequest, a model-agnostic framework for privacy-preserving LLM inference that reformulates privacy protection at the batch level rather than the individual-prompt level. The key idea is to obscure sensitive information by mixing original prompts with noisy variants, while grouping semantically equivalent instructions to amortize the inference cost over a large batch of queries with minimal impact on LLM response quality. This design is independent of the LLM architecture, requiring no access to model parameters or architectural modification. Empirical results demonstrate that SharedRequest achieves over $20\%$ higher utility compared to prior differential privacy baselines, and its shared-prompt mechanism reduces query cost by up to $5\times$ compared to non‑batched inference.
\end{abstract}

\section{Introduction}
Rapid advancement of large language models (LLMs) has enabled transformative applications in healthcare, finance, and personal assistance \cite{chen2025improving, wang2025imagent,wu2025ensemble,li2025fakeradar}. State-of-the-art public LLMs, including ChatGPT, Claude, and Gemini, are primarily deployed on the cloud platform, which protects proprietary model architectures but raises significant privacy concerns. User prompts often contain sensitive information that should remain confidential from any third party, underscoring the urgent need for effective privacy-preserving inference frameworks.

Current paradigms for privacy-preserving LLM inference are constrained by two fundamental limitations:
\begin{itemize}[noitemsep,topsep=0.2pt]
    \item First, they suffer a trilemma between privacy, utility, and efficiency. Cryptographic methods such as secure multi-party computation (SMPC) \cite{hao2022iron, luo2024secformer} typically introduce substantial communication and computation overhead, limiting their practicality for large-scale LLM deployment. Local differential privacy (LDP) approaches \cite{du2023dp, mai2024split} introduce randomization on client side, which often significantly degrade utility and are often unsuitable for text generation.
    \item Second, most existing solutions require LLM-specific architectural modifications, imposing significant burdens on service providers. On the other hand, current model-agnostic methods perturb queries individually \cite{utpala2023locally,tong2025inferdpt, chen2023customized}, substantially distorting semantic meaning and compromising response quality.
\end{itemize}

We address these bottlenecks through two critical observations in modern LLM usage:
\begin{itemize}[noitemsep,topsep=0.2pt]
    \item \textbf{Batched query processing.} With massive user bases, commercialized LLMs typically process large batches of queries within short time windows. This creates opportunities to amortize inference costs across users. 
    \item \textbf{Sparse sensitive content.} Sensitive information within user prompts is often sparse, and not all tokens require protection. For example, in the query "Provide career suggestions for a person working in cybersecurity", only the word "cybersecurity" is sensitive.
\end{itemize}

These insights suggest a \emph{batch-level} formulation of privacy-preserving LLM inference, where \emph{privacy protection operates on sensitive prompt components and the additional cost is shared across the batch}.

Motivated by this idea, we propose SharedRequest\footnote{\url{https://github.com/NusIoraPrivacy/SharedRequest}},
% \footnote{\url{https://anonymous.4open.science/r/SharedRequest-B719}}
a privacy‑preserving LLM inference paradigm. Rather than perturbing each prompt \emph{independently}, our framework groups semantically equivalent requests for joint processing at the \emph{batch level}. It then samples noisy variants of sensitive content and mixes them with original prompts to obscure private information, while amortizing the additional inference cost related to noisy queries across a large batch of requests. We further design a lightweight cryptographic protocol to ensure secure communication among multiple parties. On three privacy-sensitive benchmarks, our SharedRequest consistently outperforms prior model-agnostic LDP baselines with over $20\%$ higher utility, while the shared-request mechanism reduces query cost by up to $5.6\times$ compared with non-batched inference.

Our contributions are summarized as follows:

(1) We propose a novel \emph{batch-level} paradigm for \emph{model-agnostic} private LLM inference, where semantically equivalent requests are grouped for joint processing, enabling both privacy protection through shuffled noisy variants and cost reduction through batch-level amortization.

(2) We design a lightweight multi-party protocol that preserves user anonymity and hides private attributes from the service provider, while requiring no modification to the LLM architecture or access to model parameters. 

(3) We formalize privacy through $(A_n, \epsilon)$-indistinguishability and provide theoretical analysis of both privacy protection and query-cost amortization.

(4) We introduce a third-party coordination mechanism for efficient noisy-query sampling, reducing user-side computation for prompts with multiple private attributes.

\section{Related Work}
\subsection{SMPC-based Private Transformer Inference}
SMPC enables multiple parties to conduct transformer model inference while ensuring that the server learns nothing about user inputs \cite{yao1982protocols}. Iron \cite{hao2022iron} introduces a hybrid cryptographic protocol tailored for matrix multiplications and complex non-linear operations, including Softmax, GELU, and LayerNorm. MPCFormer \cite{limpcformer} integrates MPC with knowledge distillation to approximate expensive functions cryptographically. SecFormer \cite{luo2024secformer} builds on this foundation by replacing nonlinear operations with optimized SMPC-friendly approximations and redesigning Softmax, GeLU, and LayerNorm. BOLT \cite{pang2024bolt} reduces the payload through cryptographic optimizations, minimizing homomorphic encryption rotations via a baby-step giant-step algorithm and employs word-elimination heuristics. NEXUS \cite{zhang2024secure} introduces the first non-interactive SMPC protocol for transformer inference based on RNS‑CKKS homomorphic encryption, reducing latency and interaction overhead compared to traditional multi-round protocols. Despite their strong privacy guarantees, these SMPC-based methods entail substantial computational and communication overhead, limiting their scalability in practical deployment.

\subsection{Differentially Private Prompt Engineering}
LDP mechanism has been explored to privatize user prompts before they are sent to cloud-based LLMs. Early methods convert tokens into embeddings, add manipulated noises, and map the embeddings back, effectively preserving differential privacy \cite{lyu2020differentially, lyu2020towards}. To achieve better utility, recent approaches differentiate sensitive from non-sensitive tokens and replace sensitive tokens using constrained adjacency lists rather than full vocabulary sampling \cite{tong2025inferdpt, chen2023customized, yue2021differential,li2024feddiv}. For example, CusText assigns smaller replacement sets per token to reduce semantic drift while maintaining privacy. Other approaches include paraphrase-based mechanisms such as DP-Prompt, which produce privatized outputs via temperature-controlled sampling during text rewriting \cite{utpala2023locally,mattern2022limits}. While LDP-based prompt privatization methods can operate with black‑box models, they generally face fundamental challenges in balancing trade-offs between utility and privacy.

\begin{table}[htp]
\begin{small}
% \scalebox{0.89}{
\begin{center}
\begin{tabular}{lccc}
\hline
 &  \makecell{Text \\Generation} & \makecell{Model-\\Agnostic} & \makecell{Utility-\\Preserving} \\
 \hline
Iron &  $\times$ & $\times$ & $\checkmark$\\
MPCFormer & $\times$ & $\times$ & $\checkmark$ \\
BOLT & $\times$ & $\times$ & $\checkmark$\\
NEXUS & $\times$ & $\times$ & $\checkmark$\\
 \hline
 RanText & $\checkmark$ & $\checkmark$ & $\times$\\
 CusText & $\checkmark$ & $\checkmark$ & $\times$\\
 DP-Prompt & $\checkmark$ & $\checkmark$ & $\times$\\
 InferDPT & $\checkmark$ & $\checkmark$ & $\times$\\
  \hline
DP-Forward & $\times$ & $\times$ & $\times$\\
SnD & $\times$ & $\times$ & $\times$\\
 \hline
 \textbf{SharedRequest} & $\checkmark$ & $\checkmark$ & $\checkmark$\\
  \hline
\end{tabular}
\caption{Comparison among transformer-based privacy-preserving inference frameworks.}
\label{tab:comparison}
\end{center}
% \vskip -0.2in
% }
\end{small}
\vskip -0.1in
\end{table}

\section{Background}
\subsection{Problem Formulation}
A service provider owns a proprietary generative language model $LM: \mathcal{V}^* \rightarrow \mathcal{V}^*$ such as ChatGPT and Gemini \cite{openai2023gpt4,hurst2024gpt,team2023gemini,chen2025mark,cai2025deepshield}, where $\mathcal{V}^*$ represents the space of all token sequences over the vocabulary $\mathcal{V}$. Given a user prompt $q$, our framework aims to protect its private attributes $A_q=\{a_1, a_2, ..., a_{|A(q)|}\}$ against the service provider.

The user's prompt $q$ consists of two parts: (i) the generic instruction $T_q \in \mathcal{T}$, and (ii) the personal attributes $A_q \in \mathcal{A}$. For example, given the prompt "What are the recommended restaurants in Britain?", the personal attribute is "Britain" and the generic instruction is "What are the recommended restaurants". For any generic instruction $T\in \mathcal{T}$, we define the plausible attribute set $\mathcal{A}(T)\subset \mathcal{A}$ as those $\mathcal{A}(T)$ that co-occurs with $T$ with non-zero probability (e.g., various locations compatible with that instruction). In the following, we use $A_q$ and $A(q)$ interchangeably to denote the private attributes of prompt $q$.

\subsection{A Simple Construction}
Instead of transmitting a single prompt to the cloud LLM, we propose sending a set of prompts, containing original and noisy prompts, to obscure the user's sensitive information. 

Given a prompt $q$ with private attributes $A_q$ and generic instruction $T_q$, the user samples a set of noisy attributes $\mathcal{A}'(T_q)$ from the plausible attribute set $\mathcal{A}(T_q)$. Each sampled attribute $A'(T_q)\in \mathcal{A}'(T_q)$ is then embedded into generic instruction $T_q$ to form noisy prompts. The noisy prompts, together with the original, are submitted to the server, and only the response to the true prompt $q$ is retained locally.

While providing accurate responses and preserving privacy, such construction introduces two practical challenges. Firstly, the query cost scales by a factor of $|\mathcal{A}'(T_q)|$ compared to a non-private scenario. Secondly, for prompts containing multiple private attributes, it becomes resource-intensive for the user to generate a series of plausible combinations. In the following, we propose a framework to mitigate the query cost and user computation overhead.

\section{Framework}

\subsection{Overview}

Our protocol involves grouping prompts and sampling noisy queries through a third party. We posit that the service provider receives a large volume of requests within each time window, making it highly likely that a certain number of prompts share the same generic instruction. For example, ChatGPT processes more than 1 billion queries every day, which translates to more than 11,500 queries per second \cite{Singh_About}. By grouping requests based on their generic instructions, we can amortize query costs and server computation time. 

As depicted in Figure \ref{fig:overview}, our framework consists of three parties:
\begin{figure*}[h]
  \centering
  \includegraphics[width=0.85\linewidth]{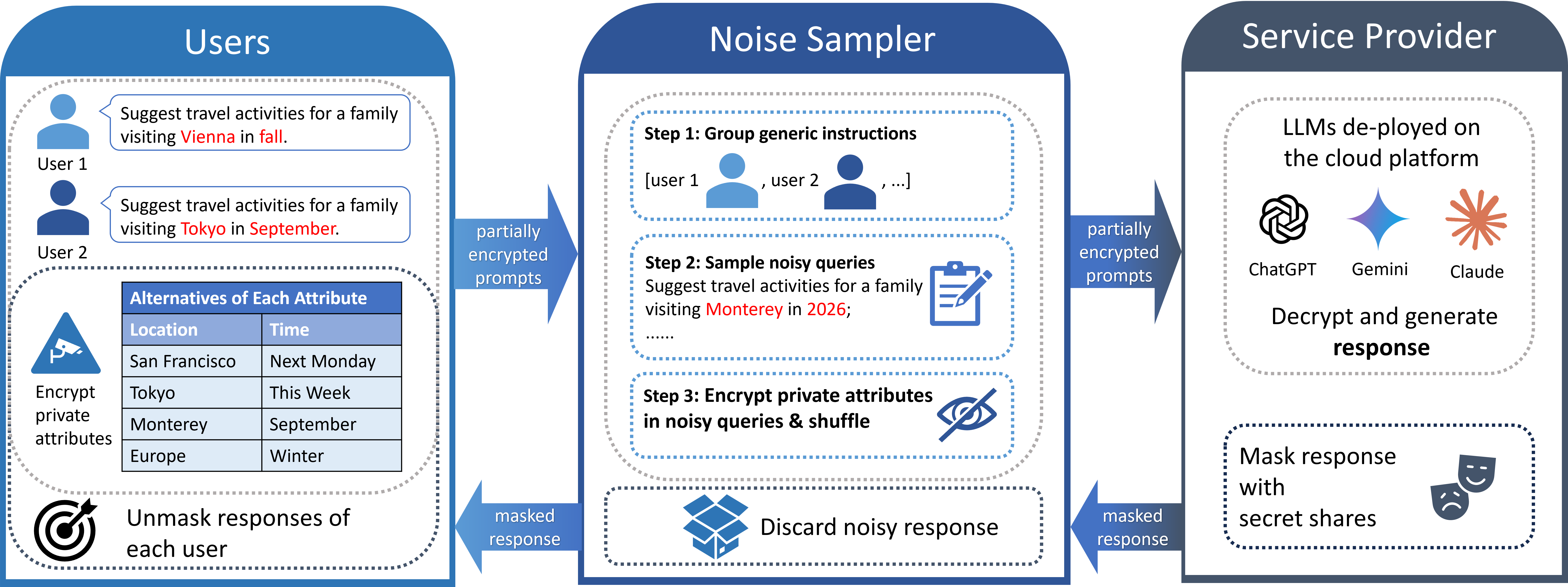}
  \caption{Overview of SharedRequest. The noise sampler groups user requests based on their generic instruction, samples noisy prompts with each group, and sends a shuffled mix to the server for inference. The inference costs related to noisy queries are amortized across users.}
  \label{fig:overview}
  \vskip -0.1in
\end{figure*}

\textbf{Users} hold their queries with private attributes to be kept hidden from the service provider. For each prompt $p$, they identify the private attributes $A_p$ and sample noisy attributes for each individual private attribute $a\in A_p$.

\textbf{Noise Sampler} receives queries with encrypted private attributes from the user. It clusters prompts with the same generic instruction and samples noisy queries to obfuscate original prompt.

\textbf{Service Provider} receives queries from the noise sampler and generates answers per query. The answers are masked with users' secrets and transmitted back to the noise sampler.

Appendix~\ref{app:algo} outlines the algorithm for our privacy-preserving inference framework.

\subsection{Cryptographic Design}
Our protocol aims to ensure that: (i) The noise sampler has no knowledge of any sensitive data and only observes the generic instruction. (ii) The service provider's view consists of an anonymous set of shuffled prompts, including genuine and noisy ones. A core design challenge lies in ensuring secure communication among users, noise sampler, and service provider, while also preserving the user’s anonymity from the service provider. 

\subsubsection{Forward Transmission}
The forward stage transmits of the message from the user to the service provider. During this phase, the user encrypts the private attributes using the service provider’s public key $pk_s$, and the service provider decrypts the ciphertext with its private key $sk_s$. This asymmetric encryption scheme hides the private attribute from any third parties, including the noise sampler. Furthermore, any messages sent to the noise sampler are encrypted with the sampler’s public key $pk_n$, preventing eavesdroppers (including the service provider) from learning their contents.

\subsubsection{Backward Transmission}
The backward stage transmits the response from the service provider to the user. The response may contain information related to private attributes, and thus it should be concealed from the noise sampler. This phase enables the user to decrypt the response without revealing their identity to the service provider. To accomplish this, we design a masking-based obfuscation scheme: (i) The user sends a seed $s$ encrypted with the service provider’s public key. (ii) The service provider decrypts the seed $s$ and uses a pseudorandom generator (PRG) to derive a pseudo-random mask $e=PRG(s)$, which is used to obfuscate the response via modular addition $r_s=r+e$. (iii) Upon receiving $r_s$, the user reconstructs the response via modular subtraction $r=r_s-e$.

\subsection{Noisy Query Sampling}

To amortize query cost, the noise sampler groups prompts with semantically equivalent generic instructions to share the same plausible attribute sets (see Appendix \ref{app:group}). When a query involves multiple private attributes, sampling plausible attribute combinations becomes significantly more complex. If there are $\mu$ private attributes, each with $k$ plausible alternatives, the total number of possible combinations grows to $k^\mu$. To address this combinatorial explosion, we introduce a coordinated approach between the user and the noise sampler, which consists of three components: (i) candidate specification for individual attributes by user; (ii) combination selection by noise sampler; and (iii) combination sampling by noise sampler (see Appendix \ref{app:dpsamp}).

\subsubsection{Candidates for Individual Attributes (User)}
Rather than sampling full attribute \textit{combinations}, the user independently selects alternative attributes for each \textit{single} attribute: $\{\mathcal{A}'_1, \mathcal{A}_2',...,\mathcal{A}'_{|A(q)|}\}$, where $|A(q)|$ is the number of attributes within prompt $q$ and $\mathcal{A}'_i$ is the set of alternatives for the $i^{th}$ attribute. These alternatives can be manually curated or drawn from a pre-constructed attribute-class database.

\subsubsection{Combination Selection (Noise Sampler)}
It would be prohibitively expensive to sample and send all combinations to the service provider, with payload exponential in the attribute size. To optimize the query cost, the noise sampler chooses combinations that look genuine within the context: (i) randomly chooses a set of candidate combinations from the users' individual attribute lists; (ii) employs a pre-trained discriminator to score each combination’s genuineness within context; (iii) selects only those combinations whose scores exceed a pre-specified threshold $\delta$, obtaining qualified combination set $\mathcal{A}^n$. 

To ensure adequate coverage, i.e., each of the $\mu$ attributes has all candidates with probability $\geq p$, the noise sampler draws $m$ combinations satisfying:
\begin{equation}
\begin{gathered}
    m \geq \left(\log (1-p) - \log (\mu k)\right) / \log (1-1/k),\\
    k = \max\{k_1, k_2, ..., k_{\mu}\},
\end{gathered} 
\end{equation}
where $k_i$ is the number of candidates for the $i^{th}$ attribute. The derivation is detailed in Appendix~\ref{app:sample}. In practice, we can set $m =\alpha \cdot \left(\log (1-p) - \log (\mu k)\right) / \log (1-1/k)$, where $\alpha$ is a tuning factor to balance privacy guarantees and system overhead.

% \section{Theoretical Analysis}
\section{Privacy Analysis}
\subsection{Threat Model}
We assume that both the service provider and noise sampler are curious-but-honest \cite{yang2019federated}. 
% The integrity and authenticity of transmitted message can be ensured by a UF-CMA secure signature scheme \cite{kaur2012digital, katz2010digital}. 
Furthermore, we assume that there is no collusion between the noise sampler and service provider. In Appendix \ref{app:shareplus}, we extend our framework using multi-server architecture and signature scheme \cite{kaur2012digital, katz2010digital} to withstand stronger adversary models, including malicious noise sampler and collusion between servers.
\subsection{Noise Sampler}
The noise sampler's view consists of: (i) each user's generic instruction $T_{q}$; (ii) each user's encrypted private attributes and seed; (iii) candidate alternatives for each individual attribute; and (iv) masked response for each user. The IND-CPA security of the encryption scheme guarantees that the sampler cannot infer any sensitive information from ciphertexts. Moreover, given that the seed used for masking is chosen uniformly at random, the security of the PRG ensures that the random mask $e$ and thus the obfuscated output $r_s$ are computationally indistinguishable from random numbers. Therefore, the noise sampler gains no information about the user prompt except its generic instruction and candidate alternatives.
\subsection{Service Provider}
The service provider's view is a shuffled set of prompts, including both genuine and noisy queries. To formulate the privacy protection, we introduce the notion of $(\mathcal{A}^n,\epsilon)$-indistinguishability. Appendix \ref{app:compprivacy} elaborates its post-processing and composition rules.

\begin{definition} [$\mathcal{A}^n$-Neighbors] 
Denote $A_q$ and $T_q$ as, respectively, the private attributes and generic instruction of user's prompt $q$. Let $\mathcal{A}^n (T_q)\subset \mathcal{A}$ be the set of qualified alternative attributes for $T_q$. Two sets of prompts, $Q'$ and $Q$, are $\mathcal{A}^n$-neighbors if $Q'$ can be obtained from $Q$ by 
removing the private attributes $A_q$ from one prompt $q$ in $Q$. We denote $Q'\in N_{\mathcal{A}^n} (Q)$.
% changing one prompt $q$ in $Q$ on its private attributes from $A_q$ to $A'_q\in \mathcal{A}^n (T_q)$. We denote $Q'\in N_{\mathcal{A}^n} (Q)$.
\end{definition}
\begin{remark}
    The private attributes $A_q$ can be removed from $q$ either by generalizing them into public information or by masking them directly.
\end{remark}

\begin{definition} [$(\mathcal{A}^n,\epsilon)$-indistinguishability] 
For any $\epsilon>0$, a randomized mechanism $M: \mathcal{Q}\rightarrow \mathcal{Y}$ preserves $(\mathcal{A}^n,\epsilon)$-indistinguishability if for $\forall Q,Q'\in N_{\mathcal{A}^n} (Q)$, the following inequality holds:
\begin{equation}
    \Pr[M(Q) \in S] \leq e^{\epsilon} \Pr[M(Q') \in S],
\end{equation}
for all subsets \(S\) of the output space $\mathcal{Y}$.
\end{definition}

\begin{theorem}
\label{thm:indist}
Denote $M: \mathcal{P} \rightarrow \mathcal{Y}$ as our protocol that maps user prompt batches to the shuffled output seen by the service provider. For any $\epsilon>0$, the mechanism $M$ achieves $(\mathcal{A}^n,\epsilon)$-indistinguishability.
\end{theorem}

This theorem ensures that the service provider’s view remains indistinguishable with regards to the change in private attributes.

\subsubsection{Connection to Differential Privacy} In the following, we show that the definition of $(\mathcal{A}^n,\epsilon)$-indistinguishability is a customization and relaxation of differential privacy (DP) \cite{dwork2006differential}.

\begin{definition} [Differential Privacy] 
Given $\epsilon>0$ and $\delta \in [0,1)$, a randomized mechanism preserves \((\epsilon, \delta)\)-differential privacy if and only if, for any neighboring sets of prompts \(Q, Q' \in \mathcal{Q}\), the following inequality holds:
\[
\Pr[M(Q) \in S] \leq e^{\epsilon} \Pr[M(Q') \in S] + \delta
\]
for all subsets \(S\) of the output space $\mathcal{Y}$.
\end{definition}
In standard DP, two datasets $Q'$ and $Q$ are considered neighbors if they differ by exactly one prompt, i.e., one query $q\in Q$ is replaced to any other prompt $q'\in Q'$. We can obtain $(\mathcal{A}^n,\epsilon)$-indistinguishability from standard DP by setting the constraints: (1) only change $q$'s private attribute $A_q$, and (2) private attributes are replaced by the prompt's qualified alternatives $\mathcal{A}^n (T_q)$. Conversely, standard DP can be obtained from $(\mathcal{A}^n,\epsilon)$-indistinguishability by treating $A_q$ as the full prompt $q$, and allowing $\mathcal{A}^n (T_q)$ to include all possible text values within the context length. Therefore, $(\mathcal{A}^n,\epsilon)$-indistinguishability can be considered as a user-defined variant of DP, with qualified alternatives and thus neighbors co-specified by users and noise sampler.

\section{Complexity Analysis}
The computation cost of our protocol is primarily determined by: (i) noise sampler's grouping and attribute combination filtering cost; (ii) optional user-side query simplification cost; (iii) service provider's LLM inference cost; (iv) encryption and decryption operations. Since (iv) is negligible relative to (i)-(iii) according to our empirical analysis, we focus on the costs related to the former two components.
\subsection{Noise Sampler Computation Cost} 

The primary computational overhead for the noise sampler arises from batch grouping and attribute filtering. The batch grouping involves: (i) encoding each generic instruction into embedding; and (ii) clustering queries based on the embeddings. The first part is fully parallelizable across instructions, and the second part depends on the chosen clustering algorithm \cite{ester1996density,sumengen2021scaling,pelleg1999accelerating}. Modern clustering algorithms, such as Reciprocal Agglomerative Clustering (RAC), enable parallel processing, scaling efficiently to billions of points. Additionally, combination filtering, which evaluates plausible attribute combinations, costs $O\left(\log(\mu k)/\log (1-1/k)\right)$ per generic instruction. In practice, this can also be parallelized across samples, further reducing running time.

\subsection{User Simplification Cost}
Prompt simplification is performed optionally and can be executed manually or via automated methods. For complexity analysis, we assume a local model deployed by the user. Consequently, the computational cost depends directly on the model’s architecture and the lengths of the original and simplified prompts.

\subsection{Query Cost from Service Provider} 

We consider query cost as the token-based charge from commercial LLM APIs, mostly associated with service provider's LLM inference cost. To mitigate this cost, the server aggregates identical prompts and generates the response once per group rather than handling each individually. The following theorem shows how query cost depends on the distribution of generic instructions.
\begin{theorem}
\label{thm:complexity}
Suppose a batch contains $B$ genuine queries, each drawn independently from an underlying distribution. Denote $c_i$ and $p_i$ as, respectively, the average query cost and sampling probability for generic instruction $T_i, i\in [|\mathcal{T}|]$. Then the expected per-query cost $C$ is:
\begin{equation}
    E(C) = \frac{1}{B}\sum_{i=1}^{|\mathcal{T}|} \left(1-(1-p_i)^B\right) c_i . 
\end{equation}
If $c_i=c$ is equal for all generic instructions $T_i, i\in [|\mathcal{T}|]$, then $E(C)$ increases with the entropy of the instruction distribution $H(T)$, and the cost maximizes at uniform distribution:
\begin{equation}
    \lim_{|\mathcal{T}|\rightarrow\infty} E(C) = \frac{|\mathcal{T}|}{B} (1-e^{-B/|\mathcal{T}|}) c.
\end{equation}
\end{theorem}
\begin{remark}
Here, $p_i$ refers to the real-world occurrence probability of generic instruction $T_i, i\in [|\mathcal{T}|]$, not related to the sampling mechanism of noisy queries. Semantically equivalent generic instructions are treated as the same $T_i$.
\end{remark}
Similarly, the cost reduction compared with non-batched scenario is associated with the entropy of instruction distribution (see Appendix \ref{app:querycost}). The above theorem implies that our sharing strategy is especially cost-effective when the query distribution is uneven, such as in long-tail scenarios. Empirical studies have observed such non-uniform distributions in real-world LLM usage \cite{kelly2025investigating}.

\section{Experiment}
\subsection{Datasets and Setting}
Our method is tested on three privacy-sensitive datasets: (i) Legal-QA that contains professional question–answer pairs covering real-world legal scenarios \cite{dzunggg_legal_qa_v1}; (ii) Medical-QA with medical problems \cite{chen2024medical}; and (iii) MMLU-Biz that extracts the business-related questions from MMLU dataset \cite{hendryckstest2021}. 
% We assess utility using the F1 score on MMLU‑Biz and employ quality score rated by GPT‑4o for both Legal‑QA and Medical‑QA.

We evaluate two variants of our framework: (i) original: all users transmit the raw prompt, and (ii) simplified: all users apply prompt simplification to improve efficiency. In practice, mixed deployments across users will likely yield utility and inference overhead between these two extremes.

\subsection{Utility Analysis}
We compare the utility of our algorithm with: (1) non-private setting where users issue prompts directly to the LLM, and (2) three model-agnostic private inference approaches (see Table \ref{tab:comparison}), including: 
% (i) non-private scenario where users issue prompts directly to the LLM without privacy safeguards; 
(i) RanText \cite{tong2025inferdpt} that introduces random adjacency list for text perturbation; (ii) CusText \cite{chen2023customized} that provides relaxed DP guarantees by customizing each input token’s replacement set; (iii) DP-Prompt \cite{utpala2023locally} that applies DP sampling during paraphrase generation via temperature-controlled decoding; (iv) CusText+ that is adapted to provide similar level of relaxed DP of SharedRequest; (v) InferDPT \cite{tong2025inferdpt} that consists of a DP-based perturbation module and an extraction module for denoising. 

To provide similar privacy guarantee with the relaxed DP version of SharedRequest, we adapt CusText+ to the aligned relaxed DP version by restricting its replacement space to the same candidate list $A_n$ and sampling replacements via an exponential mechanism. For InferDPT, we keep its extraction module unchanged, and replace its perturbation module with CusText+'s $A_n$-restricted sampling.

As shown in Table \ref{tab:utilitynp}, the original version achieves nearly identical utility to the non-private baseline (within $1\%$), where any minor loss stems from instruction clustering and is negligible. The simplified version incurs an average utility loss of approximately $4.9\%$ relative to the non-private setting.
\begin{table}[htp]

\begin{small}
\begin{center}
\scalebox{0.88}{
\begin{tabular}{lccc}
\hline
 & MMLU-Biz & Medical-QA & Legal-QA \\
 \hline
& \multicolumn{3}{c}{GPT-3.5-Turbo}\\
\cmidrule(l){2-4} 
Non-private & $\mathbf{0.671}$\tiny$\mathbf{\pm0.003}$ & $6.89$\tiny$\pm0.01$ & $7.31$\tiny$\pm0.01$ \\
Ours (Original) & $0.665$\tiny$\pm0.002$ &$\mathbf{6.91}$\tiny$\mathbf{\pm0.01}$ & $\mathbf{7.32}$\tiny$\mathbf{\pm0.02}$\\
Ours (Simplified) & $0.638$\tiny$\pm0.002$ & $6.32$\tiny$\pm0.04$ & $7.02$\tiny$\pm0.05$ \\
 \hline
& \multicolumn{3}{c}{GPT-4o-mini}\\
\cmidrule(l){2-4} 
Non-private & $\mathbf{0.853}$\tiny$\mathbf{\pm0.001}$ & $\mathbf{8.60}$\tiny$\mathbf{\pm0.01}$ & $\mathbf{8.69}$\tiny$\mathbf{\pm0.03}$ \\
Ours (Original) & $0.851$\tiny$\pm0.004$ & $8.58$\tiny$\pm0.00$ & $8.63$\tiny$\pm0.02$\\
Ours (Simplified) & $0.817$\tiny$\pm0.001$ & $8.28$\tiny$\pm0.03$ & $8.27$\tiny$\pm0.04$ \\
 \hline
& \multicolumn{3}{c}{GPT-4o}\\
\cmidrule(l){2-4} 
Non-private & $0.899$\tiny$\pm0.001$ & $\mathbf{8.81}$\tiny$\mathbf{\pm0.03}$ & $\mathbf{8.81}$\tiny$\mathbf{\pm0.01}$\\
Ours (Original) & $\mathbf{0.900}$\tiny$\mathbf{\pm0.002}$ & $8.74$\tiny$\pm0.04$ & $8.79$\tiny$\pm0.01$\\
Ours (Simplified) & $0.848$\tiny$\pm0.003$ & $8.40$\tiny$\pm0.02$ & $8.46$\tiny$\pm0.06$ \\
 \hline
\end{tabular}
}
\end{center}
\end{small}
\caption{Utility comparison of non‑private setting and our approach. We report F1 score for MMLU-Biz, and quality score rated by GPT-4o on a 1–10 scale for Medical‑QA and Legal‑QA.}
\label{tab:utilitynp}
\vskip -0.1in
\end{table}

Table \ref{tab:utilitydp} compares the utility of our simplified variant against DP baselines under privacy budget from $1$ to $5$. It can be observed that our method consistently outperforms the baselines. At $\epsilon=1$, our simplified version achieves $2.2\times$, $1.7\times$, $1.7\times$, $1.2\times$, and $1.2\times$ improved utility on average than RanText, CusText, DP-Prompt, CusText+ and InferDPT methods, respectively.

\begin{table}[htp]
\begin{small}
\begin{center}
\begin{tabular}{lccc}
\hline
Method & $\epsilon$ & GPT-4o-mini & GPT-4o\\
\hline
% \multicolumn{4}{l}{Standard DP Perturbation}\\
% \hline
\multirow{3}{*}{\shortstack[l]{RanText\\(Standard DP)}} & $1$ & $0.381$\tiny$\pm0.002$ & $0.390$\tiny$\pm0.006$\\
& $3$ & $0.388$\tiny$\pm0.006$ & $0.405$\tiny$\pm0.008$\\
& $5$ & $0.411$\tiny$\pm0.005$ & $0.435$\tiny$\pm0.006$\\
 \hline
\multirow{3}{*}{\shortstack[l]{CusText\\(Standard DP)}}&  $1$ & $0.511$\tiny$\pm0.008$ & $0.473$\tiny$\pm0.005$ \\
& $3$ & $0.553$\tiny$\pm0.002$ & $0.603$\tiny$\pm0.003$ \\
& $5$ & $0.631$\tiny$\pm0.006$ & $0.649$\tiny$\pm0.005$ \\
  \hline
\multirow{3}{*}{\shortstack[l]{DP-Prompt\\(Standard DP)}}& $1$ & $0.497$\tiny$\pm0.010$ & $0.496$\tiny$\pm0.008$\\
& $3$  & $0.522$\tiny$\pm0.009$ & $0.533$\tiny$\pm0.006$\\
& $5$  & $0.529$\tiny$\pm0.003$ & $0.542$\tiny$\pm0.003$\\
 \hline
\multirow{3}{*}{\shortstack[l]{CusText+\\(Relaxed DP)}}& $1$ & $0.686$\tiny$\pm0.002$ & $0.694$\tiny$\pm0.003$ \\
& $3$ & $0.695$\tiny$\pm0.004$ & $0.758$\tiny$\pm0.006$  \\
& $5$ & $0.743$\tiny$\pm0.005$ & $0.754$\tiny$\pm0.005$ \\
 \hline
 \multirow{3}{*}{\shortstack[l]{InferDPT\\(Relaxed DP)}}& $1$ & $0.700$\tiny$\pm0.006$ & $0.712$\tiny$\pm0.008$ \\
& $3$ & $0.708$\tiny$\pm0.003$ & $0.758$\tiny$\pm0.006$  \\
& $5$ & $0.746$\tiny$\pm0.001$ & $0.763$\tiny$\pm0.003$ \\
 \hline
%   \multicolumn{4}{l}{SharedRequest}\\
% \hline
\textbf{Ours (Simplified)} & $1$ & $\mathbf{0.817}$\tiny$\mathbf{\pm0.001}$ & $\mathbf{0.848}$\tiny$\mathbf{\pm0.003}$\\
 \hline
\end{tabular}
\end{center}
\end{small}
\caption{F1 score of our approach (with prompt simplification across all users) and DP baselines across privacy budget $\epsilon=1$ to $5$ on MMLU-Biz.}
\label{tab:utilitydp}
\vskip -0.1in
\end{table}

% \begin{table}[htp]
% \begin{small}
% \begin{center}
% \begin{tabular}{lcc}
% \hline
% Method & GPT-4o-mini & GPT-4o\\
% \hline
% % \multicolumn{4}{l}{Standard DP Perturbation}\\
% % \hline
% RanText & $0.381$\tiny$\pm0.002$ & $0.390$\tiny$\pm0.006$ \\
%  \hline
% CusText&  $0.511$\tiny$\pm0.008$ & $0.473$\tiny$\pm0.005$ \\
%   \hline
% DP-Prompt & $0.497$\tiny$\pm0.010$ & $0.496$\tiny$\pm0.008$\\
%  \hline
% CusText+ & $0.686$\tiny$\pm0.002$ & $0.694$\tiny$\pm0.003$ \\
%  \hline
% InferDPT & $0.700$\tiny$\pm0.006$ & $0.712$\tiny$\pm0.008$ \\
%  \hline
% %   \multicolumn{4}{l}{SharedRequest}\\
% % \hline
% \textbf{SharedRequest} & $\mathbf{0.817}$\tiny$\mathbf{\pm0.001}$ & $\mathbf{0.848}$\tiny$\mathbf{\pm0.003}$\\
%  \hline
% \end{tabular}
% \end{center}
% \end{small}
% \caption{F1 score of our approach (with prompt simplification across all users) and DP baselines across privacy budget $\epsilon=1$ to $5$ on MMLU-Biz.}
% \label{tab:utilitydp}
% \vskip -0.1in
% \end{table}

\subsection{Attack Results}
While the service provider only observes a shuffled set of prompts, it is crucial that the original queries remain indistinguishable from the injected noisy prompts. To assess the indistinguishability, we simulate the following inference attack: service provider utilizes a pre-trained classifier to label each prompt as genuine or noisy, aiming to distinguish true user queries. Based on the classification result, the server can identify the original prompts.

In Figure \ref{fig:pjdatk}, we perform experiments using two Qwen-based discriminators, i.e., Qwen2.5-0.5B and Qwen2.5-1.5B \cite{qwen2025qwen25technicalreport}, to generate candidate combinations. For each discriminator, we train four attack classifiers of varying model sizes \cite{liu2019roberta, qwen2025qwen25technicalreport, dubey2024llama} to distinguish genuine prompts from noisy variants. Depending on the sampled candidate set size $m$, the number of plausible combinations ranges from $8.3$ to $52.4$, and each attribute has over $7$ alternatives on average for $\alpha\geq 10$. Among the qualified attributes, the F1 score of the attack is within $63\%$ for Qwen2.5‑0.5B and $58\%$ for Qwen2.5‑1.5B discriminators. In contrast, when the server randomly samples fake prompts without filtration, the attack success rate is much higher, reaching around $80\%$. Our combination-filtering mechanism reduces attack success by an average of approximately $32.7\%$, substantially improving robustness against server inference attacks. 

\begin{figure}[htp]
    \centering
    \includegraphics[width=0.98\linewidth]{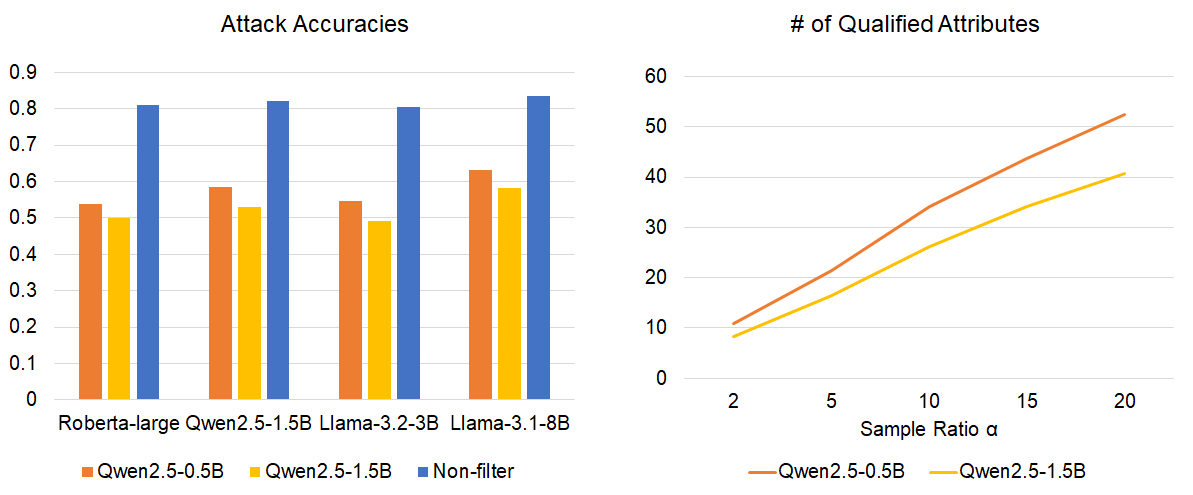}
    \caption{AUC of attack accuracies and size of qualified attributes on MMLU-Biz dataset using Qwen2.5-0.5B and Qwen2.5-1.5B discriminators. In the “non‑filter” scenario, the server randomly samples fake prompts without applying combination filtering.}
    \label{fig:pjdatk}
\vskip -0.1in
\end{figure}

In Figure \ref{fig:atccomp}, we evaluate empirical privacy under attribute inference attack, where uses GPT-5.2 to infer whether a sensitive attribute is owned by the user on MMLU-Biz dataset. It can be observed that SharedRequest achieves comparable attack ASR relative to DP-Prompt, CusText+, and InferDPT, while maintaining substantially higher utility. While SharedRequest leads to higher attack ASR for RanText and CusText by up to $20\%$ and $10\%$, respectively, the F1 score is higher than the two baselines by at over $80\%$.

\begin{figure}[htp]
    \centering
    \includegraphics[width=0.98\linewidth]{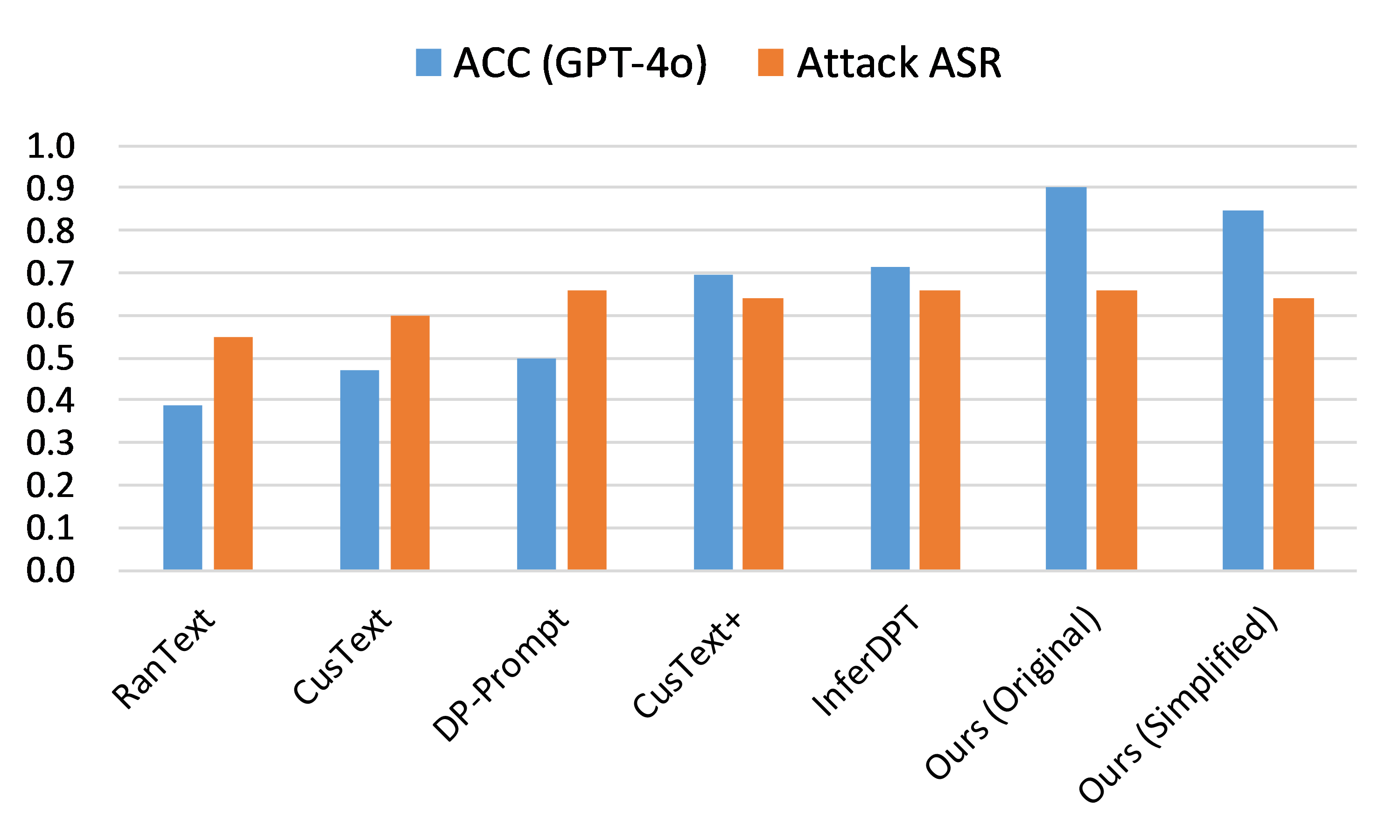}
    \caption{Utility (F1 score) and ASR for attribute inference attack of SharedRequest and the baselines under $\epsilon=1$}
    \label{fig:atccomp}
\vskip -0.1in
\end{figure}

\subsection{Query Cost}
We compare monetary query cost of our protocol against a non-batched baseline, where noisy queries are not shared across users. To simulate real-world request distribution, we draw prompts from a Dirichlet distribution with concentration parameter $\beta$, where larger $\beta$ reflects higher uniformity. As is shown in Figure \ref{fig:querycost}, more concentrated distributions (smaller $\beta$) yield greater cost savings under our request sharing mechanism. At $\beta=0.05$, our protocol reduces the monetary cost by up to $5.6\times$ compared to the non-batched setting. Prompt simplification further enhances this effect, since it increases collisions in generic instructions and amplifies batching efficiency.

\begin{figure}[htp]
    \centering
    \includegraphics[width=0.98\linewidth]{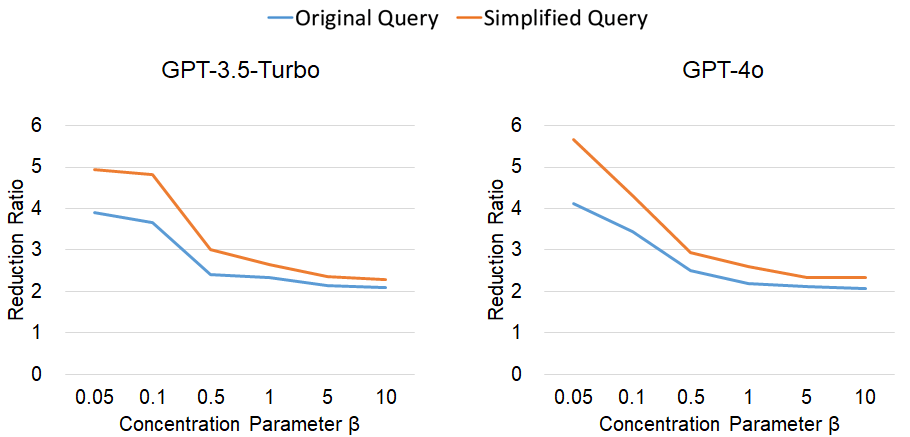}
\caption{Query cost reduction under varying concentration parameters $\beta$ for gpt models. Reduction ratio is measured as the ratio of non-batched baseline cost to that of SharedRequest.}
\label{fig:querycost}
\vskip -0.2in
\end{figure}

\subsection{Computation Cost for Attribute Filtering}

Figure \ref{fig:genmmlu} reports the computation cost required for attribute combination filtering. Computation time increases roughly linearly with sample ratio $\alpha$, staying below $2.1$s for Qwen2.5‑0.5B and $6.2$s for Qwen2.5‑1.5B when $\alpha\leq 20$. Overall, Qwen2.5‑1.5B incurs about a $2.8\times$ overhead compared to Qwen2.5‑0.5B, suggesting the trade-off between computation overhead and privacy protection. Prompt simplification further improves filtering efficiency, reducing runtime by an average of $40.8\%$, with larger efficiency gains at higher $\alpha$. 

\begin{figure}[htp]
    \centering
    \includegraphics[width=0.98\linewidth]{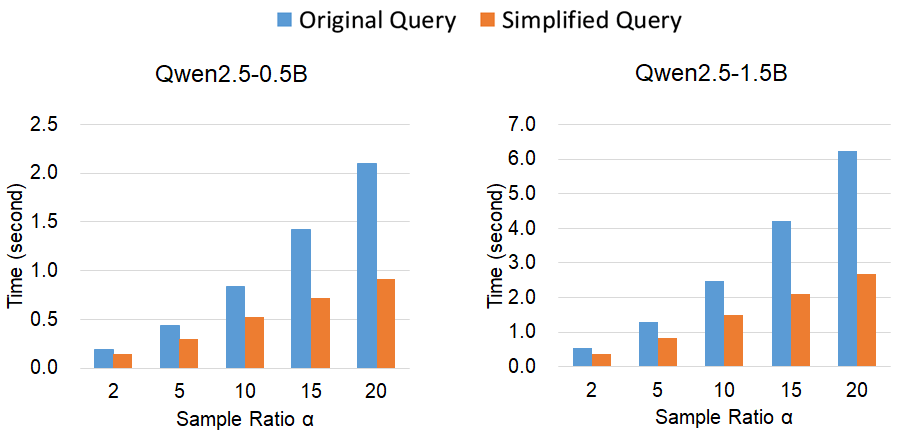}
\caption{Computation cost per generic instruction group (in seconds) for attribute filtering on MMLU-Biz using Qwen2.5‑0.5B and Qwen2.5‑1.5B discriminators.}
\label{fig:genmmlu}
\vskip -0.2in
\end{figure}

\section{Conclusion}
This paper proposes SharedRequest, a privacy-preserving LLM inference paradigm that requires no architectural modifications to the model. Our multi-party protocol obscures user prompts by injecting noisy variants, where the additional inference costs are shared among a large batch of users. To alleviate user-side computation cost, we design a coordinated mechanism between users and noise sampler to sample qualified noisy queries. Theoretical analysis reveals that SharedRequest achieves greater query-cost reduction when request distributions are uneven, such as long-tailed. The empirical evaluation demonstrates that: (1) SharedRequest incurs minor utility degradation compared to non-private baseline, and outperforms DP baselines by over $20\%$. (2) Under concentrated distributions, SharedRequest reduces query costs by up to $5.6\times$ versus non-batched baselines. Further discussions on our framework is provided in Appendix \ref{app:discuss}.

\section*{Limitation}
\label{app:limit}
Our work has the following limitations and practical considerations.

\textbf{Assumption of non-colluding parties.}
Our protocol assumes that the noise sampler can not collude with the service provider to ensure user privacy. In practice, the noise sampler can be: (1) a third party who provides the privacy-related service; or (2) an organization such as finance company that protects the confidentiality of employees’ or clients’ prompts. While two non-colluding parties is a common assumption in SMPC protocols \cite{corrigan2017prio, mohassel2017secureml}, the security could be enhanced by increasing the colluding threshold. In Appendix \ref{app:shareplus}, we extend our protocol to multi-server setting, where an adversary needs to collude with all $m$ noise samplers and the service provider to compromise privacy.

\textbf{Private attributes and alternatives.}
We assume that users are responsible for defining their own privacy, i.e., identifying private attributes and generating alternatives for each attribute. This user‑defined privacy paradigm is increasingly recognized in academic literature and implemented in practice \cite{van2014privacy, busch2019implementing}. Appendix \ref{app:privattr} discusses methods for the self-definition process, along with empirical results.
In practice, users can add an anonymization layer prior to SharedRequest, removing information that is irrelevant to the task yet increases privacy risk. For example, they might might pattern‑match and redact phone numbers or replace user names with generic placeholders like [NAME]. Then SharedRequest can be implemented to protect attributes that are sensitive while relevant for LLM inference.

\section*{Ethical Considerations}
We emphasize that all datasets used in this work are publicly available benchmarks widely adopted in the NLP community for evaluation, such as standard language understanding and task benchmarks. These datasets are intended for research and model comparison and do not contain privately collected or user-generated content with personally identifying information. We also confirm that the data used in our experiments do not include offensive or sensitive content linked to identifiable individuals.

\section*{Acknowledgements}

The authors utilized ChatGPT for polishing writings of the manuscript. All technical content, experimental design, analysis, and conclusions were developed and verified by the authors.

This work was supported by the Chongqing Key Laboratory of Trusted Perception and Interaction Technology for Intelligent and Connected Vehicles, the State Key Laboratory of Intelligent Vehicle Safety Technology, Chongqing Changan Automobile Co., Ltd, and the Chongqing Natural Science Foundation (Grant No. CSTB2024NSCQ-LZX0172). Ding was supported in part by a grant from the Israel Science Foundation (ISF Grant No. 1774/20), and by the European Union (ERC, SCALE,101162665). Views and opinions expressed are
however those of the author(s) only and do not necessarily reflect those of the
European Union or the European Research Council. Neither the European Union
nor the granting authority can be held responsible for them.

\bibliography{custom}

\appendix

\section{Algorithm}
\label{app:algo}
Algorithm \ref{alg:batchforward} and \ref{alg:batchbackward} outline our SharedRequest, consisting of forward and backward transmission.
\begin{algorithm}[htp]
   \caption{SharedRequest: Forward Transmission}
   \label{alg:batchforward}
\begin{algorithmic}
   \STATE \textbf{Input:} User queries $Q=\{q_1, q_2,...,q_N\}$.
   \STATE \textbf{Output:} Responses $R=\{r_1, r_2,...,r_S\}$, where response $r_i$ corresponds to original or noisy query $q_i$.
   % \STATE
   % \STATE \textbf{Forward Transmission}
   \STATE \textbf{User} $i\in [N]$
   \STATE \textbullet~ \textit{Sample noisy attributes}: The user identifies private attributes $A_{q_i}$, and samples alternatives for each individual attributes
   $\mathcal{A}^i=\{\mathcal{A}_1^i, \mathcal{A}_2^i,...,\mathcal{A}_{|A(q)|}^i\}$.
   \STATE \textbullet~ \textit{Encrypt attribute and seed}: The user encrypts the private attributes $A_{q_i}$ and sampled seed $s$ with service provider's public key, such that $Enc_s^i=Enc(pk_s, s||A_{q_i})$.
   \STATE \textbullet~ \textit{Encrypt and submit message}: The user encrypts the message with the noise sampler's public keys: $Enc_n^i=Enc(pk_n, T_{q_i}||\mathcal{A}^i||Enc_s^i)$, and sends the encrypted message $Enc_n^i$ to noise sampler.
   \STATE
   \STATE \textbf{Noise Sampler} $n$
   \STATE \textbullet~ \textit{Decrypt messages}: For every received cipher $Enc_n^i$, the sampler decrypts the cipher with its secret key $sk_n$, producing $T_{q_i}$, $\mathcal{A}^i$, and $Enc_s^i$.
   \STATE \textbullet~ \textit{Batch prompts}: The sampler groups prompts based on their generic instruction via a batching module $\mathcal{B}$, producing $K$ groups $[G_1, G_2, ..., G_K]$ and their corresponding generic instruction $[T_{G_1}, T_{G_2}, ..., T_{G_k}]$.
   \STATE \textbullet~ \textit{Sample noisy query}: For each group $G_k$, the sampler: (i) unions all alternatives for each single attribute; (ii) samples and filters attribute combinations via a local discriminator, producing $\mathcal{A}^n(T_{G_k})$; (iii) randomly samples $\mathcal{N}[v]$ attribute combinations.
   \STATE \textbullet~ \textit{Encrypt attribute and seed}: Given the total count of dummy queries $S=\sum \mathcal{N}$, the sampler randomly generates $S$ seeds. Then it encrypts the noisy attributes and seeds with the service provider's public key.
   \STATE \textbullet~ \textit{Shuffle and forward message}: The sampler shuffles all messages and sends them to the service provider $s$.
    \STATE
   \STATE \textbf{Service Provider} $s$
   \STATE \textbullet~ \textit{Decrypt message}: For every received message $Enc_s^i$, the service provider decrypts the cipher with its secret key $sk_n$, producing $s^i$ and $A_{p^i}$.  
   \STATE \textbullet~ \textit{Generate response}: The service provider recovers each prompt by inserting the private attributes into the generic instruction $T_{G_k}, k\in [K]$. Then it groups the same prompts and generates the response for each group at once.
\end{algorithmic}
\end{algorithm}

\begin{algorithm}[h!]
   \caption{SharedRequest: Backward Transmission}
   \label{alg:batchbackward}
\begin{algorithmic}
   \STATE \textbf{Input:} Responses for a set of queries (including original and noisy ones) $R=\{r_1, r_2,...,r_S\}$.
   \STATE \textbf{Output:} Responses $R=\{r_1, r_2,...,r_N\}$, where response $r_i$ corresponds to original query $q_i$.
   % $e=PRG(s)$, which is used to obfuscate the response via modular addition $r_s=r+e$. (iii) On receiving $r_s$, the user reconstructs the response via modular subtraction $r=r_s-e$
    \STATE \textbf{Service Provider} $s$
   \STATE \textbullet~ \textit{Handle and forward responses}: For each response $r_i$, the service provider generates random number from the corresponding seed $e_i^s=PRG(s_i)$, computes $r_s^i=r_i + e_i^s$, and sends them to noise sampler $n$.
   \STATE 
   \STATE \textbf{Noise Sampler} $n$
   \STATE \textbullet~ \textit{Handle responses}: On receiving the response from service provider $s$, the sampler de-shuffles them using the inverse permutation, and discards the response corresponding to the noisy queries.
   \STATE \textbullet~ \textit{Forward responses}: The sampler sends all responses $r_i^s$ to the corresponding user $i$.
   \STATE 
   \STATE \textbf{User} $i\in [N]$
   \STATE \textbullet~ \textit{Reconstruct response}: The user receives response $r_i^s$ from noise sampler $n$, generates random number from their seed $e_i^s=PRG(s_i)$, and reconstructs $r^i=r_i^i-e_i^s$.
\end{algorithmic}
\end{algorithm}

\section{Grouping and Prompt Simplification}
\label{app:group}
\subsection{Grouping}
To amortize query cost, the noise sampler groups prompts with semantically equivalent generic instructions to share the same plausible attribute sets. We leverage a sentence transformer to encode each instruction into embedding, and apply clustering to group similar embeddings into batches automatically. To improve clustering efficiency and downstream utility, the sampler first pre-groups prompts based on the number of sensitive attributes, enabling parallelized clustering within each subgroup.

% Details of the empirical implementation are provided in Appendix \ref{app:batch}.

\subsection{Prompt Simplification}
Users may optionally apply prompt simplification to improve efficiency by: (i) reducing private attributes, which lowers the sampling and query cost; (ii) standardizing generic instructions, enabling the plausible attributes to be amortized among more queries. Simplification can be achieved through manual prompt engineering or automated tools, such as lightweight local models.

\section{Combination Sampling by Noise Sampler}
\label{app:dpsamp}
Directly combining noisy attributes with real ones can inadvertently expose genuine queries. If a genuine attribute appears multiple times in a batch while a fake one appears only once, an adversarial server would distinguish real queries by frequency patterns. Since private attributes are encrypted, the noise sampler can not guarantee that noisy attributes will match the frequency of genuine ones.

To mitigate this issue, we propose a sampling method based on one-sided exponential mechanism \cite{rozanov2012probability, kozubowski2000multivariate, takagi2022asymmetric}. Instead of injecting each qualified combination exactly once, we add randomized counts $\mathcal{N}[v]$ for each attributes combination $v\in \mathcal{A}^n$. Specifically, for each $v$, the noise sampler draws $\lambda_v$ from the one-sided exponential distribution:
\begin{equation}
   f(\lambda_v)= \begin{cases}\frac{\epsilon}{2}\exp{\frac{-\lambda_v \epsilon}{2}} & \lambda_v \geq 0,\\
    0 & \lambda_v < 0,
    \end{cases}
\end{equation}
with privacy parameter $\epsilon$. The sampler then injects $\mathcal{N}[v]=\lfloor \lambda_v \rfloor$ instances of each combination $v\in \mathcal{A}^n$. Even though each query might appear multiple times, the query cost maintains the same as the service provider groups identical prompts and generates a single response per group for payload optimization (see Appendix \ref{app:dedup}).

\section{Sampling Size}
\label{app:sample}
Denote $k_i$ as the number of candidates for the $i^{th}$ attribute, and $\mu$ is the number of private attributes within a group. The probability that an alternative attribute is not covered in a sample is:
\begin{equation}
    \left(1-1/k_i\right)^m,
\end{equation}
for attribute $i$.

Using union bound, the probability that each of the $\mu$ attributes has all candidates is bounded by:
\begin{equation}
    \sum_{i}^{\mu} k_i\left(1-1/k_i\right)^m \leq \mu k\left(1-1/k\right)^m \leq 1-p,
\end{equation}
where $k = \max\{k_1, k_2, ..., k_{\mu}\}$. Taking the logarithm of both side, we have:
\begin{equation}
    m \geq \frac{\log (1-p) - \log(\mu k)}{\log (1-1/k)}.
\end{equation}

\section{Privacy Analysis}
% \subsection{Connection to Differential Privacy}
% \label{app:connectdp}
\subsection{Resistant to Stronger Adversaries}
\label{app:shareplus}
Figure \ref{fig:extend} presents the overview of our extended version.
\begin{figure*}[!htp]
  \centering
  \includegraphics[width=0.98\linewidth]{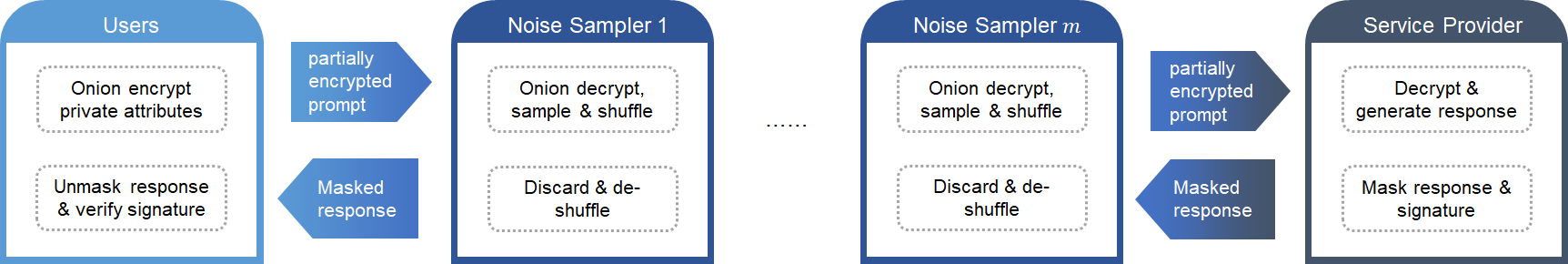}
  \caption{Overview of SharedRequest's extension. The requests are transmitted through $m$ noise samplers. On receiving response, the user unmask the message and verify the signature generated by the service provider.}
  \label{fig:extend}
\end{figure*}

\subsubsection{Collusion between Servers} Our protocol assumes that the noise sampler can not collaborate with service provider. To relax the assumption, we extend our SharedRequest to a multi-server system with $m$ noise samplers. In particular, the following steps in SharedRequest are modified:
\begin{itemize}
    \item \textbf{Onion encryption:} User $i$ onion encrypts the private attribute and seed with the servers public key: 
\begin{equation}
\begin{gathered}
    OEnc_0^i\\
    =Enc(pk_n^1, Enc(...Enc(pk_s, s||A_{q_i}))),
\end{gathered}
\end{equation}
 where $pk_n^j$ is the public key for the $j^{th}$ noise sampler. Noise sampler $j$ onion encrypts the private attribute within noisy prompt:
\begin{equation}
\begin{gathered}
    OEnc_j^i\\
    =Enc(pk_n^{j+1}, Enc(...Enc(pk_s, s||A_{q_i}))).
\end{gathered}
\end{equation}
    \item \textbf{Onion decryption:} Noise sampler $j$ onion decrypts the private attribute and seed with its private key:
\begin{equation}
    OEnc_j^i=Dec(sk_n^j, OEnc_{j-1}^i),
\end{equation}
where $sk_n^j$ is the secret key for the $j^{th}$ noise sampler. The service provider decrypts with its private key $sk_s$:
\begin{equation}
    (s,A_{q_i})=Dec(sk_s, OEnc_{m}^i).
\end{equation}
    \item \textbf{Noisy query sampling.} Each noise sampler generates qualified attribute combinations, and samples from the attributes added to the mix. 
\end{itemize}

\subsubsection{Malicious Noise Samplers} We consider the case where some noise samplers might temper with the transmitted message and return incorrect result to the user, including user queries and service provider's response. To ensure the integrity of message, the service provider creates signature for each query:
\begin{equation}
    \sigma_i = Sign(sk_s^d, q_i||r_i),
\end{equation}
where $sk_s^d$ is the private signing key. Then the user can verify the correctness of the response by asserting that 
\begin{equation}
Verify(pk_s^d, q_i||r_i, \sigma_i)=1.
\end{equation}

\subsubsection{Security Analysis} 
We provide a high-level security analysis of the extended protocol. For secrecy, the views of service provider and noise sampler are similar to that in SharedRequest. The private attributes are onion encrypted through the public keys of $m$ noise samplers and finally the service provider. An adversary would need to compromise all $m$ samplers and the service provider to decrypt the attributes. If at least one noise sampler remains honest, the privacy of user attributes is preserved. Furthermore, each sampler independently shuffles and injects noisy variants to the batch, ensuring that the service provider's view is a mix of noisy and original prompts.

For integrity, the UF-CMA secure signature scheme \cite{kaur2012digital, katz2010digital} prevents the noise samplers from modifying user queries and service provider's responses. Consequently, users are able to identify that the output is incorrect and reject it.

\subsection{Proof of Theorem \ref{thm:indist}}
\label{app:proofprivacy}
\begin{proof}
In our algorithm, modifying the private attributes within a prompt does not alter the batching outcome; it only affects the server's observed frequency of each attribute combination within a batch. Therefore, we can prove theorem \ref{thm:indist} by showing that the frequency of the private attributes, $H$, satisfies $(\mathcal{A}^n,\epsilon)$-indistinguishability.

Changing $Q$ to $Q'$, i.e., changing the private attribute in one prompt, corresponds to decreasing one element in $H$ by 1 and increasing another element in $H$ by 1. Let $f:\mathcal{Q}\rightarrow \mathcal{H}$ be the function that maps the group of prompts to the frequency before injecting noisy query, and $M':\mathcal{Q}\rightarrow \mathcal{H}$ be the mechanism that maps to the frequency after adding $\lambda$. Then we have:

\begin{equation}
\begin{gathered}
\frac{p(M'(Q)=z)}{p(M'(Q')=z)} = \frac{\prod_i \exp{\left(-\epsilon(z_i-f(Q)_i)/2\right)}}{\prod_i \exp{\left(-\epsilon(z_i-f(Q')_i)/2\right)}} \\
\leq\exp{\left(\sum_i\epsilon|f(Q')_i-f(Q)_i|/2\right)} \leq \exp{(\epsilon)}.
\end{gathered}
\end{equation}
% for any $z$ with $z_i\geq 0$, where $z_i$ represents the frequency count of the $i^{th}$ attribute. If $z_i<0$ for any $i$, then:
% \begin{equation}
% \frac{p(f(Q)=z)}{p(f(Q')=z)}=0\leq \exp{(\epsilon)}.
% \end{equation}

Therefore, for any subset $S\in \mathcal{H}$, it holds that:
\begin{equation}
\begin{gathered}
P(f(Q)\in S) \leq \exp{(\epsilon)} P(f(Q')\in S).
\end{gathered}
\end{equation}

Taking the floor of $\lambda$ is equivalent to taking the floor of frequency after adding $\lambda$. Then we have:
\begin{equation}
\begin{gathered}
    P(\lfloor f(Q) \rfloor=y) = P\left(f(Q)\in floor^{-1}(z)\right) \\
    \leq   \exp{(\epsilon)}  P\left(\left(f(Q')\in floor^{-1}(z)\right)\right)\\
    = \exp{(\epsilon)} P(\lfloor f(Q') \rfloor=y),
\end{gathered}
\end{equation}
where $floor^{-1}$ denotes the inverse of floor function. Hence we complete the proof.

\end{proof}

\subsection{Preservation under Post-processing and Composition}
\label{app:compprivacy}
Proposition \ref{prop:postprocess} and \ref{prop:comp} provides the privacy preservation under post-processing and composition.
\begin{proposition}
\label{prop:postprocess}
Let $M: \mathcal{Q}\rightarrow \mathcal{Y}$ be a $(\mathcal{A}^n,\epsilon)$-indistinguishability mechanism. Let $f: \mathcal{Y} \rightarrow \mathcal{Z}$ be an arbitrary (possibly randomized) mapping. Then $f\circ M: \mathcal{Q}\rightarrow \mathcal{Z}$ is $(\mathcal{A}^n,\epsilon)$-indistinguishability.
\end{proposition}
\begin{proof}
    Fix pair of neighboring input $Q$, $Q'\in N_{\mathcal{A}^n} (Q)$, and fix a a measurable set $S \subset \mathcal{Z}$. By the definition of $f$, we have:
    \begin{equation}
    \begin{gathered}
        \Pr[f(M(Q)) \in S] \\
        =\int_y \Pr[f(y) \in S|y] \cdot \Pr[M(Q)=y] dy \\
        \leq \int_y \Pr[f(y) \in S|y] \cdot e^{\epsilon}\Pr[M(Q')=y] dy\\
        =e^{\epsilon}\Pr[f(M(Q')) \in S].
    \end{gathered}
    \end{equation}
\end{proof}

\begin{proposition}
\label{prop:comp}
Let $M_i: \mathcal{Q}\rightarrow \mathcal{Y}_i$ be a $(\mathcal{A}^n,\epsilon_i)$-indistinguishability mechanism. Then if $\underline{M}$ is defined by:
\begin{equation}
    \underline{M}(Q) = (M_1(Q), M_2(Q),\cdots,M_k(Q)),
\end{equation}
\underline{M} satisfies $(\mathcal{A}^n,\epsilon)$-indistinguishability.
\end{proposition}
\begin{proof}
    We start with the proof for $k=2$. Let $Q$, $Q'\in N_{\mathcal{A}^n} (Q)$ be the $\mathcal{A}^n$-neighbors, and let $S=S_1\times S_2$ be any measurable set in the product output space. As $M_1$, $M_1$ have independent randomness, we have:
    \begin{equation}
    \begin{gathered}
        \Pr[(M_1(Q), M_2(Q))\in S] \\
        = \Pr[M_1(Q)\in S_1] \cdot \Pr[M_2(Q)\in S_2]\\
        \leq e^{\epsilon_1}\Pr[M_1(Q')\in S_1] \cdot e^{\epsilon_2}\Pr[M_2(Q')\in S_2] \\
        =e^{\epsilon_1+\epsilon_2} \Pr[(M_1(Q'), M_2(Q'))\in S].
    \end{gathered}
    \end{equation}
    The result can extend inductively to any $k>2$.
\end{proof}

\section{Query Cost Analysis}
\label{app:querycost}
\subsection{Proof of Theorem \ref{thm:complexity}}
\begin{proof}
For each instruction group, it appears in the batch for at least one time with probability $1-(1-p_i)^B$. Then the expected per-query cost is as follows:
\begin{equation}
    E(C) = \frac{1}{B}\sum_{i=1}^{|\mathcal{T}|} \left(1-(1-p_i)^B\right) c_i. 
\end{equation}
Thus we focus on the proof for the remaining part under $c_i=c$ for all generic instructions. To prove the relationship between $E(C)$ and $H(T)$, we start with a basic case: $p_i$ increases with $\delta$, and $p_j$ decreases with $\delta$ for any pair $p_i, p_j$.

For $E(C)$, its derivative with regards to $\delta$ is given by:
\begin{equation}
\begin{gathered}
    \frac{\nabla E(C)}{\nabla \delta} =  \frac{\nabla}{\nabla\delta} \\
    \left(\frac{c}{B}\left(\left(1-(1-p_i-\delta)^B\right)  +\left(1-(1-p_i+\delta)^B\right)\right)\right)\\
    = c\left((1-p_i-\delta)^{B-1}-(1-p_j+\delta)^{B-1}\right).
\end{gathered}
\end{equation}
Then $E(C)$ increases in $\delta$ if $\delta \leq (p_j-p_i)/2$, and decreases in $\delta$ otherwise.

For $H(T)$, it derivative is as follows:
\begin{equation}
\begin{gathered}
    \frac{\nabla H(T)}{\nabla \delta}=\frac{\nabla}{\nabla\delta} \\
    \left((p_i+\delta)\log (p_i+\delta)+(p_j-\delta)\log (p_j-\delta)\right)\\
    =\log \left(\frac{p_i+\delta}{p_j-\delta}\right).
\end{gathered}
\end{equation}
$H(T)$ also increases in $\delta$ if $\delta \leq (p_j-p_i)/2$, and decreases in $\delta$ otherwise. Therefore, the change in $\delta$ leads to the same direction of change for $E(C)$ and $H(T)$.

We can generalize the above case to any change from $\mathbf{p}=(p_1,...,p_{|\mathcal{T}|})$ to $\mathbf{p}'=(p_1',...,p_{|\mathcal{T}|}')$, which can be achieved through $|\mathcal{T}|-1$ steps of the following change (for the $i^{th}$ steps):
\begin{itemize}
    \item Add $\delta=(p_i'-p_i+p_{i-1}'-p_{i-1}+\cdots+p_1'-p_1)$ to $p_i$. 
    \item Subtract the same $\delta$ from $p_{i+1}$. 
\end{itemize}

For each step, the direction of change for $E(C)$ and $H(T)$ is the same, and thus $E(C)$ and $H(T)$ both increase or decrease with regards to the instruction distribution $\mathbf{p}$.

Given that $E(C)$ is positively association with $H(T)$, it maximize at the uniform distribution:
\begin{equation}
\begin{gathered}
    E(C)=\frac{1}{B}\sum_{i=1}^{|\mathcal{T}|} \left(1-(1-1/|\mathcal{T}|)^B\right)c.
\end{gathered}
\end{equation}
Taking the limit:
\begin{equation}
    \lim_{|\mathcal{T}|\rightarrow\infty} E(C) = \frac{|\mathcal{T}|}{B} (1-e^{-B/|\mathcal{T}|}) c.
\end{equation}
\end{proof}

\subsection{Comparison with Non-batched Scenario}
For non-batched scenario, we consider the simple construction where each user sends the noisy and original prompts to the server individually, without the noise sampler for aggregated sampling and shuffling. Denote $c_j'$ and $p_j'$ as, respectively, the average query cost and sampling probability for query $q_j, j\in [|\mathcal{Q}|]$. The cost reduction ratio can be formulated as:
\begin{equation}
\begin{gathered}
    \frac{E(C)}{E(C')} = \frac{\sum_{i=1}^{|\mathcal{T}|} \left(1-(1-p_i)^B\right) c_i}{B \sum_{j=1}^{|\mathcal{Q}|} p_j' c_j'},
\end{gathered}
\end{equation}
where $E(C')$ is the expected per-query cost for non-batched scenario.

Noted that $c_i$ and $c_j'$ is determined by the token length and number of qualified alternatives for each query or instruction group. For simplicity, we assume that $c_i=c_j'=c$ for all $i\in [|\mathcal{T}|], j\in [|\mathcal{Q}|]$. This is reasonable since the per-instruction overhead in our SharedRequest is similar to the average per-query overhead of the non‑batched baseline. Then we have:
\begin{equation}
\begin{gathered}
    \frac{E(C)}{E(C')}  = \frac{\sum_{i=1}^{|\mathcal{T}|} \left(1-(1-p_i)^B\right)}{B} = \frac{E(C)}{c}.
\end{gathered}
\end{equation}
Therefore, the reduction ratio increases with $E(C)$ and $H(T)$, maximizing at the uniform distribution:
\begin{equation}
    \lim_{|\mathcal{T}|\rightarrow\infty} \frac{E(C)}{E(C')} = \frac{|\mathcal{T}|}{B} (1-e^{-B/|\mathcal{T}|}).
\end{equation}
\section{Experiment}

\subsection{Experiment Setting}
For prompt grouping, we use paraphrase-distilroberta-base-v1 model \cite{reimers2019sentence} to convert generic instruction into embeddings, and leverage RAC \cite{sumengen2021scaling} to cluster these embeddings. We set sampling ratio $\alpha=10$, privacy parameter $\epsilon=1$, and batch size $B=5000$ unless specified. We employ RSA cryptosystem \cite{rivest1978method} for asymmetric encryption. The experiments are conducted on a $96$-core Ubuntu Linux 20.04 server with $128$GB RAM and 2 A100 driver. Each reported experiment result is an average of $3$ experiments.
\subsubsection{Prompt Simplification}
To train simplification model, We finetune Qwen3-4B \cite{yang2025qwen3} for $5$ epochs with a learning rate of $0.00001$ using AdamW optimizer~\cite{iclr/LoshchilovH19}, using demonstrations generated by GPT-4o. Specifically, we ask GPT-4o to simplify a question from the three datasets following the rules:
\begin{itemize}
    \item Remove all unnecessary or redundant information.
    \item Keep all important information necessary to answer the question.
    \item Do not use abbreviations or emojis.
    \item Compress the origin as short as you can. 
\end{itemize}
The initial version of simplification demonstrations may lose important information. To ensure quality, we iteratively refine the dataset. If the simplified version yields lower response quality compared to the original, we prompt GPT‑4o with the original prompt, initial simplification, and response evaluation to generate an improved version. The refinement refinement loop runs for $5$ rounds.

For simplified version of SharedRequest, we finetune the simplification model on $80\%$ of the constructed dataset, and implement the simplification on the remaining $20\%$ for evaluation.
\subsubsection{Cluster Algorithm}
RAC \cite{sumengen2021scaling} is employed to cluster the sentence embeddings during prompt grouping. Clustering is based on Euclidean distance between embeddings, with a merging threshold of $0.3$, and executed using $8$ parallel threads. After clustering, a single instruction is randomly sampled from each cluster to serve as the shared generic instruction for that group.

\subsubsection{Output Arrangement}
To prevent the noise sampler or any external observer from inferring private information via side channels in response lengths and timings, the service provider pads all outputs within a cluster to the maximum length observed in that cluster using masked values. Masked responses from each cluster are then delivered to the noise sampler at a constant rate of 20 tokens/s.

\subsubsection{Discrimination Model}
We tested two Qwen-based discriminators, i.e., Qwen2.5-0.5B and Qwen2.5-1.5B \cite{qwen2025qwen25technicalreport}. We input the noisy and original queries, as well as their binary labels, to train the discriminator for $5$ epochs with a learning rate of $0.00001$ using AdamW optimizer, with balanced batches containing noise and original queries in a 1:1 ratio. We maintain strict separation between the datasets used for training the discriminator and those used for attack models to prevent any data leakage.

\subsubsection{Utility Evaluation Metrics}
MMLU-Biz, which consists of multiple-choice questions, is evaluated using F1 score. Medical‑QA and Legal‑QA, which comprise of open-ended questions, are scored from 1 to 10 using GPT‑4o as an automatic judge, following \cite{zheng2023judging}.

\subsubsection{Latency Analysis}
We focus on the active processing time for latency analysis. We assume the noise sampler parallelizes API calls and cryptographic operations across each prompt, and parallelizes attribute combination filtering across each generic instruction. As we cannot directly deploy proprietary GPT models, we estimate the LLM inference time using the average query time per prompt without privacy protection. We use Qwen2.5-0.5B for combination filtering.

\subsection{Data Processing}
\subsubsection{Dataset Construction} 
The MMLU-Biz dataset is constructed by filtering questions from the business or finance-related categories: business ethics, econometrics, marketing, high school macroeconomics, high school microeconomics, management, and marketing. The final MMLU-Biz dataset contains $1534$ samples. To ensure equal representation across domains (business, legal, and medical), we then randomly sample 2,000 items each from the Legal‑QA and Medical‑QA datasets. The final evaluation set therefore consists of balanced sample sizes across the three domains. In experiment, we sample the items from Dirichlet distribution with concentration parameter $\beta=1$.

\subsubsection{Private Attributes}
The private attributes for each sample are labeled with GPT-4o using the prompt in Figure \ref{fig:privattr}. The GPT-extracted attributes are then manually checked and corrected to ensure quality.
\begin{figure*}[!htp]
\begin{small}
\begin{tcolorbox}[
                    % width=\linewidth,
                  %%enhanced,
                  %%frame hidden,
                  boxsep=2pt,
                  left=2pt,
                  right=2pt,
                  top=2pt,
                  % float=htpb!,
                  ]%%
  \textbf{User Prompt:}

  Please extract all words or phrases in the question that indicate sensitive attributes.

  The sensitive attributes include: Personal identifiable information, job, race, religion, ethnicity, religion, beliefs, age, contact information, geolocation data, residency and citizenship status, date, festival, sexual orientation or practices, gender identity, third-person pronoun, disability, employment and income details, company, physical health and mental health issues, clinical signs including medical conditions and physical examination data, behavioral data, name of medicine, including specific drugs, drug classes, and general medication types (e.g., nasal decongestants, antihistamines, beta-blockers, etc), financial data, vulnerable financial status, account details, loan information, tax records, contract, trading algorithm, confidential business information, educational records, voting status, membership in a trade union, social media and digital footprint, legal consultation and case status, Law enforcement jurisdiction and procedures, consumer rights, ownership details, criminal record or name, biometric data, genetic information, legal proceedings and disputes, surveillance and monitoring data, military and security clearance.

  Your task is to strictly extract phrases that refer to these sensitive attributes. Even if a term indirectly refers to a sensitive attribute (e.g., a drug category instead of a specific medicine name), it should be included in the list. The list of sensitive attributes provided is not exhaustive — if you encounter any word or phrase that could reasonably be considered sensitive information under privacy laws or common data protection standards (like GDPR, HIPAA, or similar frameworks), include it in the list. Prioritize any data that could identify a person, describe their personal circumstances, or reveal confidential, medical, financial, or legal information.

  Please preserve the original words from the user's question to form the attribute list and do not convert full names to abbreviations or abbreviations to full names, and do not add any extra words. Retain all duplicates, regardless of form variations (e.g., burglarize, burglarized, burglarizing).
  
Please strictly return a list of phrases in the format of ["attribute 1", "attribute 2", ..., "attribute n"]. If there is no sensitive attribute, return an empty list []. Try to identify as much as phrases as possible.

Question: \{\textit{question}\}

Sensitive attributes: 
\end{tcolorbox}
\end{small}
\caption{The GPT-4o instructions used for private attribute extraction.}
\label{fig:privattr}
\end{figure*}
\subsubsection{Alternative Attributes}
To sample alternatives for each single attributes, we prompt GPT-4o to generate candidates using the instruction in Figure \ref{fig:fake}. Each attribute has around $10$ candidates on average.

\begin{figure*}[!htp]
% \begin{small}
\begin{tcolorbox}[
                    % width=\linewidth,
                  %%enhanced,
                  %%frame hidden,
                  boxsep=2pt,
                  left=2pt,
                  right=2pt,
                  top=2pt,
                  % float=htpb!,
                  ]%%
  \textbf{User Prompt:}
  
You are an AI designed to generate accurate, contextually relevant, and trustworthy fake sensitive attributes. 
Your task is to generate one list of fake sensitive attributes for each original attribute provided in filtered private attributes based on the context of the 'question'. 
The generated attributes must meet the following criteria:

1. Numerical Handling:

    - If the original attribute is primarily numerical (e.g., percentages, durations, counts), generate fake numerical attributes that have broader variance to reduce guessability.
    
    - Avoid clustering numbers closely around the original value.
    
    - Avoid adding unnecessary descriptive context (e.g., avoid turning "50\%" into "55\% of alternative components"). Instead, generate standalone numerical fakes like "20\%", "90\%", or "10\%".
    
    - Do not introduce numerical values if the original attribute is non-numerical in any form. (e.g. aviod turning "internship period" into "3 months")

    2. Length Consistency: Ensure that the generated fake attributes have a similar length (in words or characters) to the original attribute to maintain fluency and reduce guessability.

    3. Sound Sensitive and Relaiable:

    - The fake attributes should resemble private, confidential, or sensitive concepts, particularly in legal, medical, or financial domains."
    
    - Generated fake attributes actually exist in real-world literature.

    4. No Rephrasings or Simple Synonyms: Avoid generating superficial rephrasings or synonyms (e.g., do not turn "80 years" into "eight decades"):
    
    - Create attributes that keep contextual depth and believability.
    
    - Do not generate fakes by slightly modifying words.
    
    - The genearted word should under the same category of the original attribute but with strictly distinct meanings.
    
    5. Diversity and Independence:
    
    - Generate independent fake attributes for each prompt and do not reference or rely on fake attributes generated for the other prompt. Avoid generating fake attributes that are too similar to each other.
    
    - Ensure that the fake attributes are diverse and do not share common themes or patterns.
    
    - if the input private attributes are dependent based on the context, generate fake attributes that are dependent on the context as well.
    
    6. Contextual Fluency:
    
    - Make sure that the fake attributes fit naturally and coherently within the sentence structure for each context.
    
    The list must be structured as follows:
    
    % - The first attribute must be the original attribute.
    
    - The next five attributes must be unique, fake attributes.
    
    - Output format for each list is: [["fake attribute 1\_1", ..., "fake attribute 1\_n"],...]
    
    The input is structured as follows:
    
    Private Attributes: \{\textit{private attributes}\}
    
    Question: \{\textit{question}\} 
    
    Fake attributes for question:
\end{tcolorbox}
% \end{small}
\caption{The GPT-4o instructions used for alternative attribute generation.}
\label{fig:fake}
\end{figure*}

% \subsection{Comparison with DP Baselines}
% \label{app:computility}
% To provide similar privacy guarantee with the relax DP version of SharedRequest, we adapt CusText+ by replacing only the private attributes from the user-defined candidate lists using exponential sampling mechanism. Then the algorithm satisfies $\epsilon$-DP with similar neighborhood of SharedRequest. The full results under $\epsilon$ from $1$ to $5$ is presented in Table \ref{tab:utilitydpfull}. Even under the highest privacy budget, our method consistently outperforms the baselines.

\subsection{Computation Cost of Combination Filtering}
Figure \ref{fig:gencostall} presents the computation time for attribute combination filtering across datasets. Legal-QA incurs the highest computation cost due to its longer questions and more private attributes, under $6$s for Qwen2.5-0.5B and $15$s for Qwen2.5-1.5B for $\alpha\leq20$. Under $\alpha\leq10$, the computation cost stays within $3$s for Qwen2.5-0.5B and $6$s for Qwen2.5-1.5B. Prompt simplification consistently reduces the computation cost by an average of $32.7\%$, $63.1\%$, and $53.3\%$ for MMLU-Biz, Legal-QA, and Medical-QA, respectively.
\begin{figure*}[h!]
    \centering
    \includegraphics[width=0.98\linewidth]{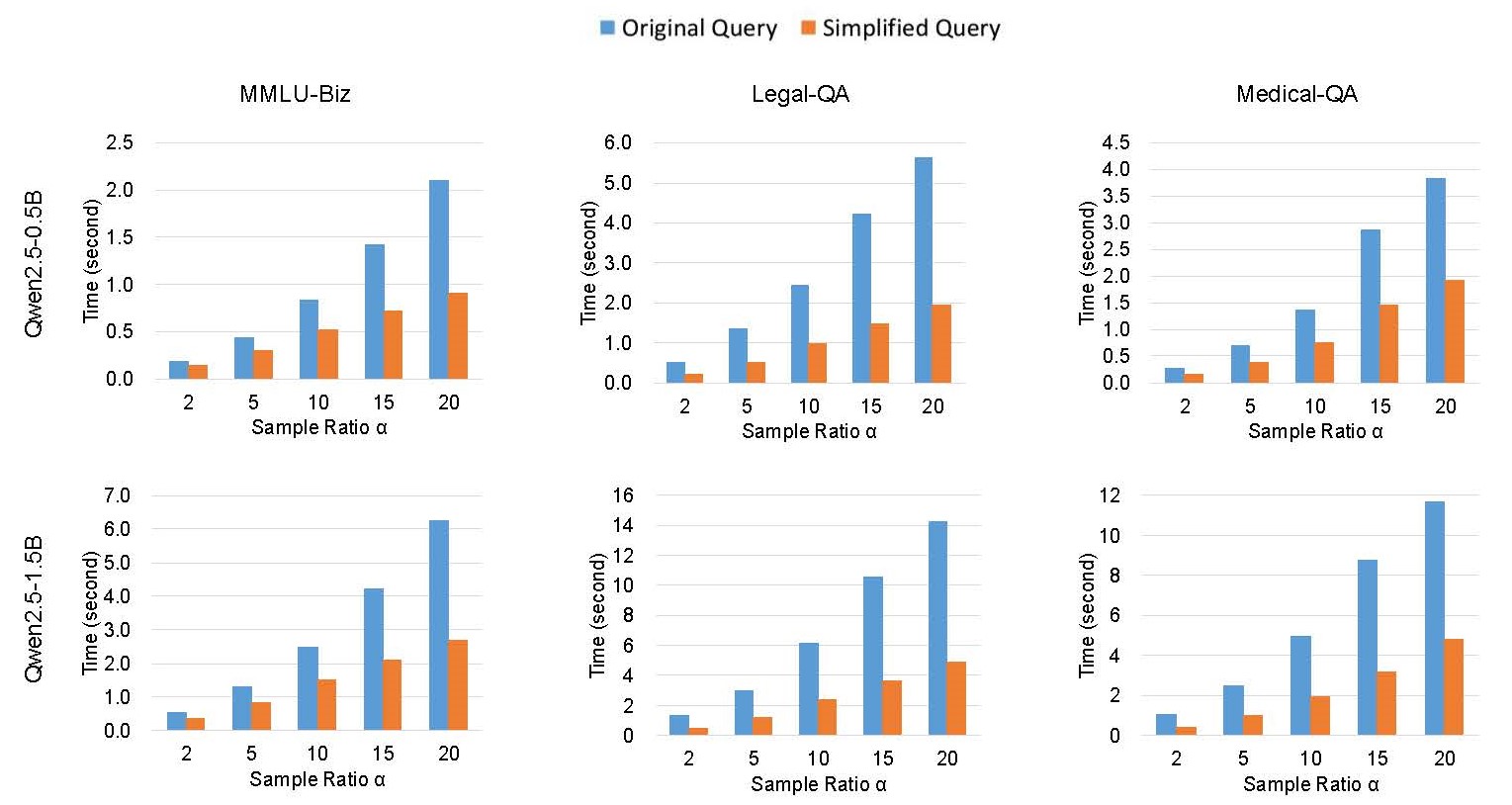}
\caption{Computation cost per generic instruction group (in seconds) for attribute filtering on three datasets using Qwen2.5‑0.5B and Qwen2.5‑1.5B discriminators.}
\label{fig:gencostall}
\end{figure*}

% \subsection{Attribute Inference Attack}

\subsection{Statistics Analysis of Private Attributes}

We summarize the statistics of private attributes and their noisy alternatives in Table \ref{tab:privstat}. We can observe that Legal-QA and Medical-QA have more private attributes within each prompt, leading to higher number of noisy combinations. Furthermore, our filtering algorithm ensures that each single attribute has over $8$ alternatives on average.
\begin{table}[h!]
\begin{center}
% \begin{sc}
\scalebox{0.9}{
\begin{small}
\begin{tabular}{lccc}
\hline
 & \makecell{\# of Private\\ Attributes} & \makecell{\# of Noisy\\ Combinations} & \makecell{\# of Alternatives\\ for Single Attribute}\\
\hline
 & \multicolumn{3}{c}{MMLU-Biz} \\
\cmidrule(l){2-4} 
Original & 3.4 & 49.4 & 9.1 \\
Simplified & 2.3 & 34.2 & 8.6 \\
\hline
& \multicolumn{3}{c}{Legal-QA} \\
\cmidrule(l){2-4} 
Original & 10.7 & 118.7 & 11.1 \\
Simplified & 5.6 & 100.2 & 11.8 \\
\hline
& \multicolumn{3}{c}{Medical-QA} \\
\cmidrule(l){2-4} 
Original & 7.9 & 62.4 & 9.4\\
Simplified &5.4 & 52.4 & 8.7\\
\hline
\end{tabular}
\end{small}
}
\end{center}
\caption{Statistics of private attribute for each dataset. The values are summarized as the average over all samples, using Qwen2.5-0.5B discriminator under sample ratio $\alpha=10$.}
\label{tab:privstat}
\end{table}

\subsection{Query Cost}
We examine the query cost under privacy parameter $\epsilon$ ranging from $0.01$ to $10$ in Figure \ref{fig:querycosteps}. While the cost decreases as $\epsilon$ increases (i.e., privacy is weakened), the query cost and thus reduction ratio retain stable under $\epsilon\leq 0.1$. This happens because the service provider aggregates and de-duplicates identical prompt variants for LLM inference, and thus the query cost is upper bounded by the size of qualified attribute combinations.

\begin{figure}[h!]
    \centering
    \includegraphics[width=0.98\linewidth]{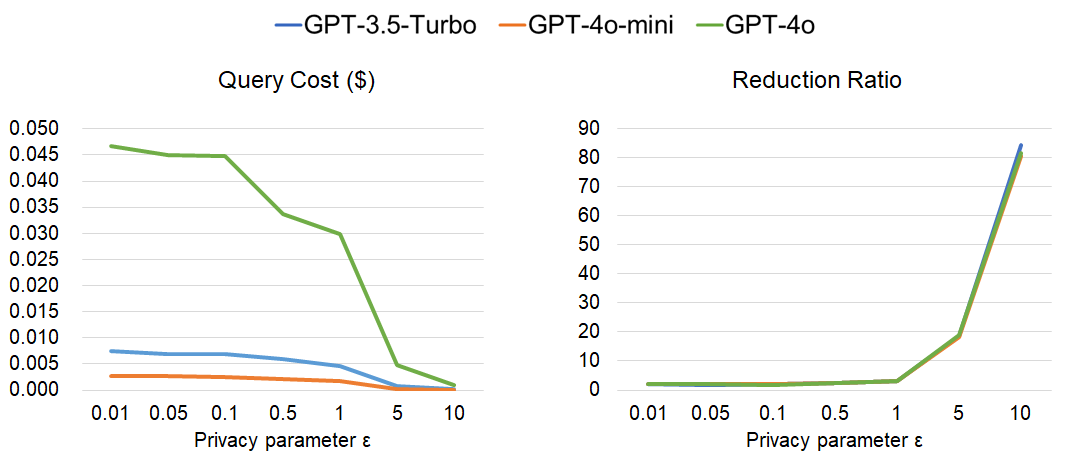}
\caption{Per-request query cost (in dollars) and reduction ratio under varying privacy parameter $\epsilon$.}
\label{fig:querycosteps}
\end{figure}

Figure \ref{fig:querycostbatch} presents the query cost under batch size $B$ ranging from $100$ to $5000$. Smaller batch size leads to lower level of reduction ratio, since the cost is distributed among fewer queries. At $B=5000$, the reduction ratio ranges from $2\times$ to $5.6\times$, whereas at $B=100$, the ratio drops to around $1.6\times$ to $2.1\times$. Prompt simplification consistently reduces the query cost across all batch configurations.

\begin{figure*}[h!]
    \centering
    \includegraphics[width=0.9\linewidth]{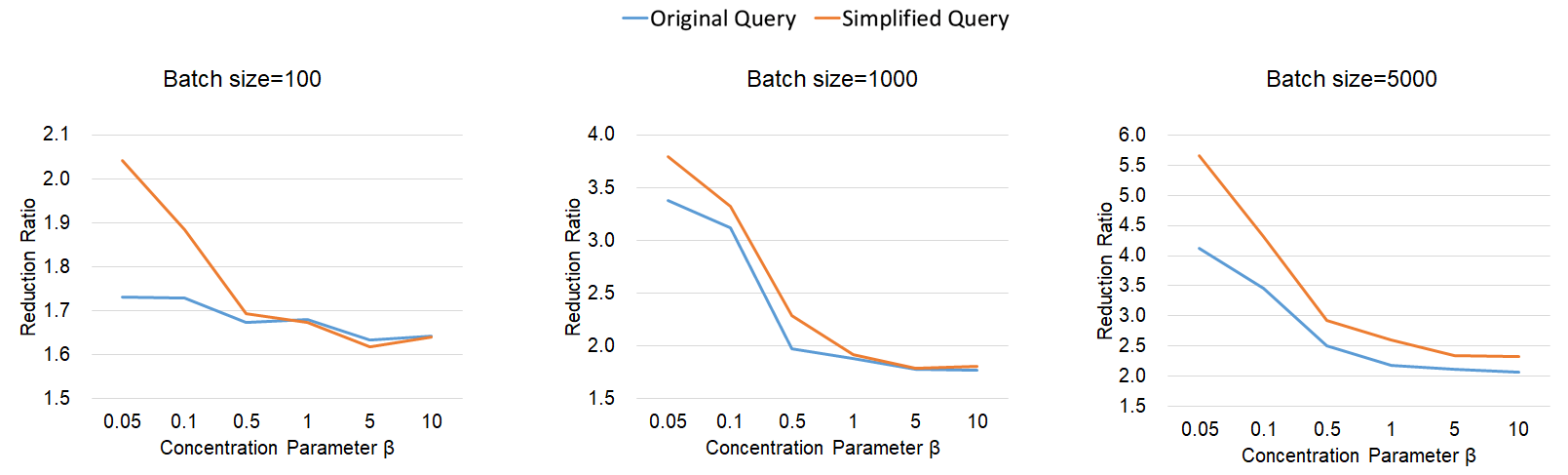}
\caption{Query cost reduction ratio under varying batch sizes $B$.}
\label{fig:querycostbatch}
\end{figure*}

\subsection{Latency Analysis} 
The latency can be decomposed into two parts: the idle waiting period used to gather prompts from multiple clients, and the active processing time required to handle the collected batch. The former component depends on the configured window time. We focus on the latter component, which comprises: (i) LLM inference time on the service provider; (ii) attribute combination filtering time for the noise sampler; (iii) prompt grouping time for the noise sampler; (iv) optional user-side query simplification time; (v) remaining overhead, mainly encryption and decryption operations. Components (ii)–(iv) represent the additional latency introduced by SharedRequest’s privacy-preserving protocol.

Figure \ref{fig:delay} breaks down active batch processing latency  under parallelized implementation. It can be observed that attribute combination filtering adds the largest overhead beyond the server-side LLM inference. While prompt simplification significantly reduces downstream processing time, it introduces an average $1.4$s of user-side overhead, illustrating a trade-off between local computation and subsequent latency. In contrast, the overheads for cryptographic operations and prompt grouping are minimal, accounting for only $0.6\%$ and $2\%$ of total latency on average, respectively.

\begin{figure}[htp]
    \centering
    \includegraphics[width=0.98\linewidth]{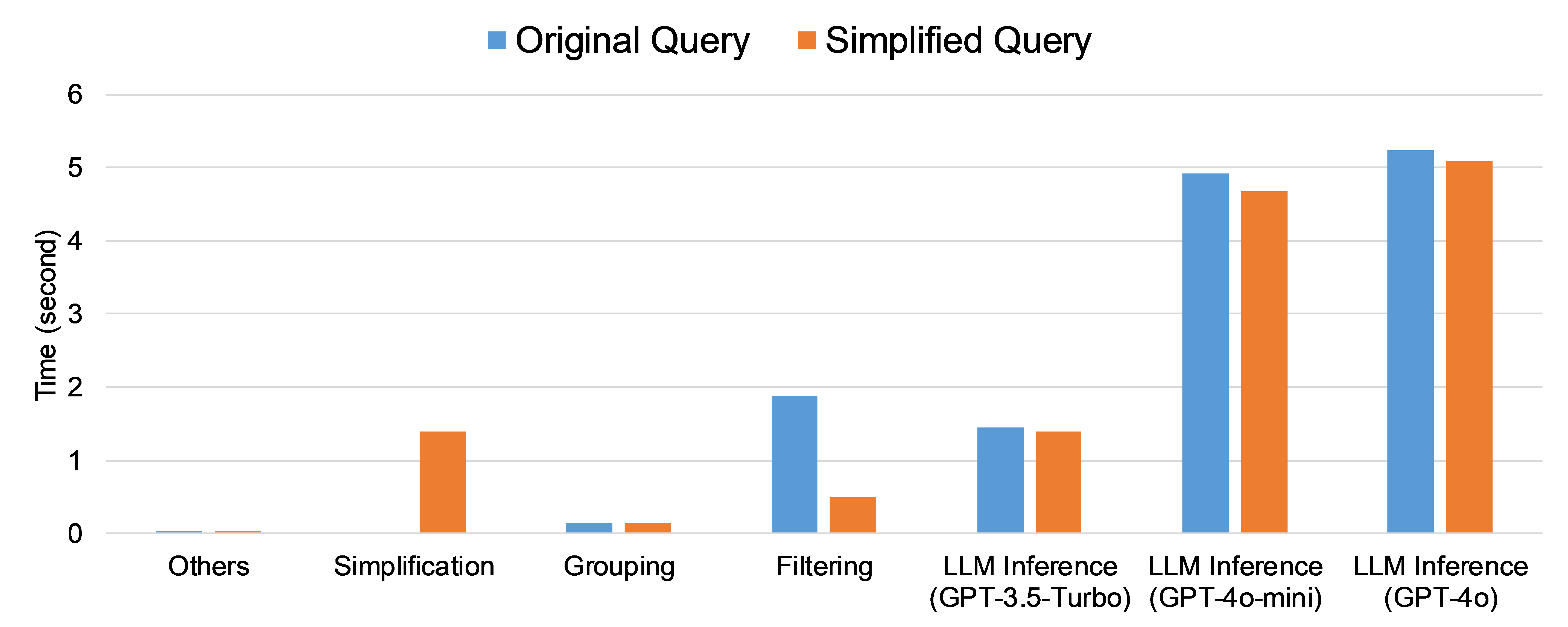}
\caption{Computation time (in seconds) per query for each process. Simplification refers to local query simplification, grouping refers to prompt grouping, and filtering refers to filtering of qualified attribute combinations.}
\label{fig:delay}
\end{figure}

\subsection{Impact of Clustering Parameter}
We study the impact of merging threshold on query costs and accuracies in Figure \ref{fig:f1thd} and \ref{fig:costthd}. As the merging threshold increases, more generic instructions are grouped into fewer clusters, which reduces both accuracy and query cost. The optimal threshold should be chosen based on the desired trade-off between utility and overhead.
\begin{figure}[htp]
    \centering
    \includegraphics[width=0.98\linewidth]{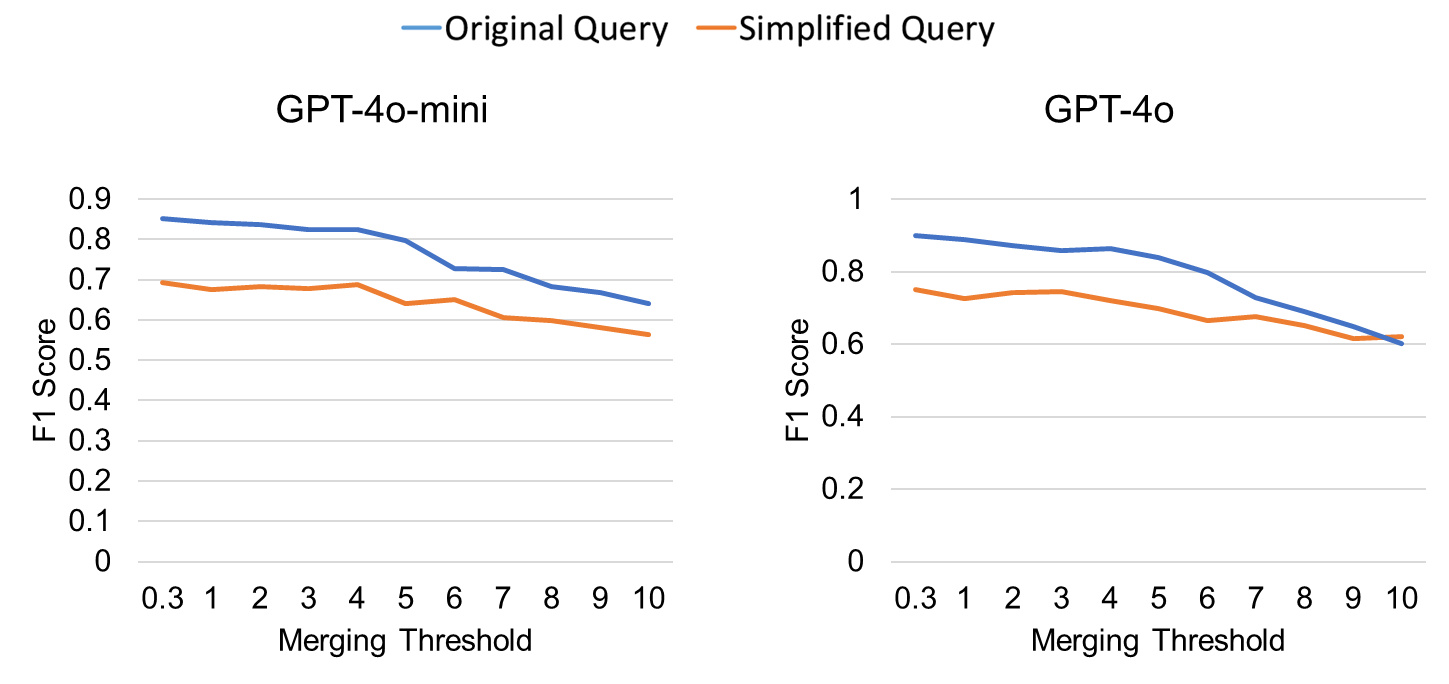}
\caption{F1 score under varying merging threshold on MMLU-Biz.}
\label{fig:f1thd}
\end{figure}

\begin{figure}[htp]
    \centering
    \includegraphics[width=0.98\linewidth]{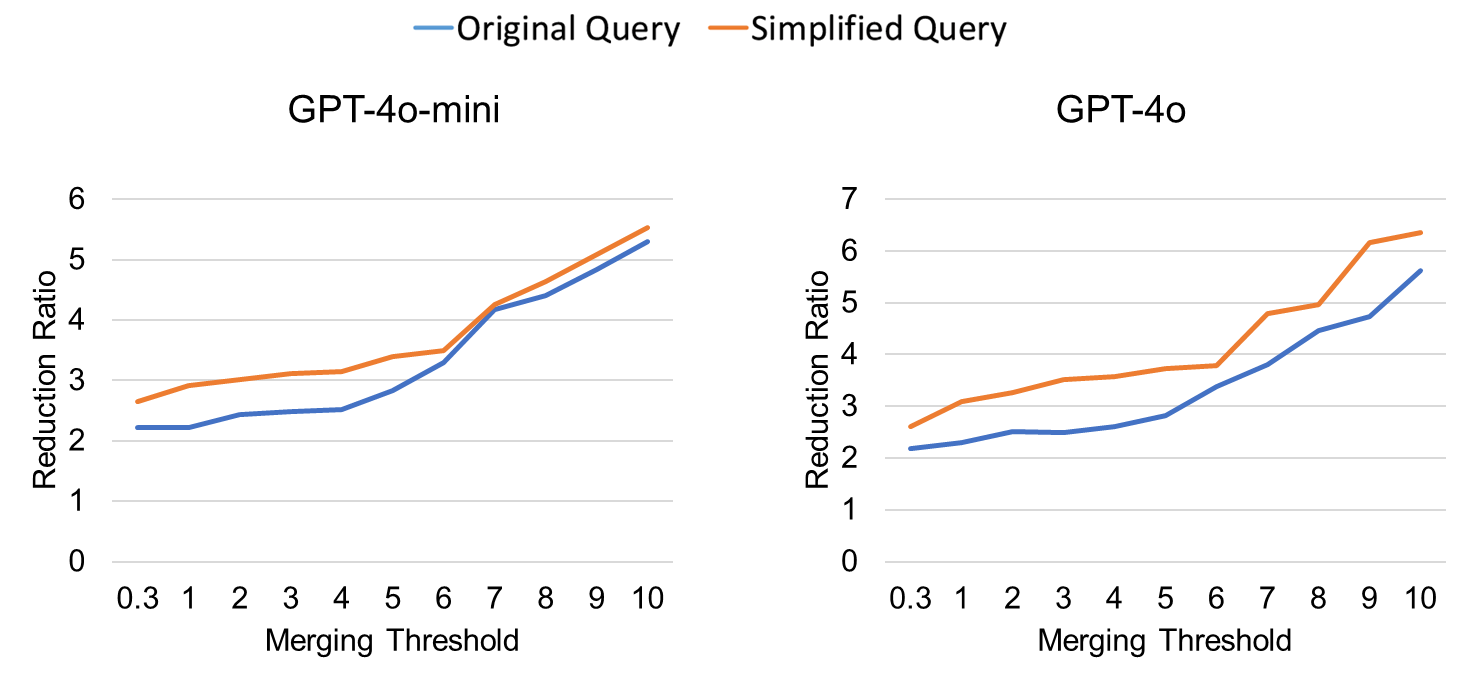}
\caption{Query cost reduction ratio under varying merging threshold on MMLU-Biz.}
\label{fig:costthd}
\end{figure}

\subsection{Overhead Comparison}
We compare the computation and communication overhead of SharedRequest with existing LLM inference frameworks. We evaluate overhead against two method types: (1) SMPC-based methods (Iron, MPCFormer, BOLT, and NEXUS) and (2) DP-based methods (RanText, CusText and DP-Prompt).

Table \ref{tab:comppriv} summarizes the resource requirements and overhead of these inference methods. SMPC-based approaches, due to their multiparty computation protocols, incur significantly higher communication overhead than both SharedRequest and DP-based methods. Moreover, despite running on smaller models (i.e., Bert-base), SMPC-based methods exhibit higher inference time compared to SharedRequest, which operates on GPT-4o. To our best knowledge, existing SMPC-based techniques have only been evaluated on models up to 13B parameters, and their scalability to larger models (e.g., 70B+) remains an open challenge. While SharedRequest has a higher runtime than perturbation-based methods, its substantial utility gains highlight its advantage.

\begin{table*}[htp]
\begin{center}
% \begin{sc}
% \scalebox{0.88}{
\begin{small}
\begin{tabular}{llcccccc}
\hline
& Framework & Model & I./O. size & Communication cost & Communication setting & Runtime \\
\hline
\multirow{4}{*}{SMPC}& Iron  & Bert-base & (128, 1) & 280 GB & (3 Gbps, 0.8 ms) & {\raise.17ex\hbox{$\scriptstyle\sim$}}475 s \\
& MPCFormer  & Bert-base & (128, 1) & 12 GB & (5 Gbps, 1 ms) & {\raise.17ex\hbox{$\scriptstyle\sim$}}55 s\\
 &BOLT  & Bert-base & (128, 1) & 25 GB & (3 Gbps, 0.8 ms) & {\raise.17ex\hbox{$\scriptstyle\sim$}}185 s \\
& NEXUS  & Bert-base & (128, 1) & 0.16 GB & (100 Mbps, 80 ms) & {\raise.17ex\hbox{$\scriptstyle\sim$}}55 s \\
\hline
\multirow{3}{*}{DP} & RanText & GPT-4o & (512, $\ast$) & $<$ 0.1 MB & (100 Mbps, 80 ms) & {\raise.17ex\hbox{$\scriptstyle\sim$}}7 s \\
& CusText & GPT-4o & (512, $\ast$) & $<$ 0.1 MB & (100 Mbps, 80 ms) & {\raise.17ex\hbox{$\scriptstyle\sim$}}6.5 s \\
& DP-Prompt & GPT-4o & (512, $\ast$) & $<$ 0.1 MB & (100 Mbps, 80 ms) & {\raise.17ex\hbox{$\scriptstyle\sim$}}7 s\\
\hline
\textbf{Ours} & \textbf{SharedRequest} & GPT-4o & (512, $\ast$) & $<$ 0.1 MB & (100 Mbps, 80 ms) & {\raise.17ex\hbox{$\scriptstyle\sim$}}10 s \\
\hline
\end{tabular}
% }
\end{small}
\end{center}
\caption{Resource requirement, communication cost, and runtime of privacy-preserving LLM inference frameworks. Values in I./O. size refers to the lengths of input and output tokens, where $\ast$ denotes an unfixed length. Communication setting specifies the network's bandwidth and latency parameters. }
\label{tab:comppriv}
\end{table*}

\section{Discussions}
\label{app:discuss}
\subsection{Private Attributes}
\label{app:privattr}
Users could identify private attributes and their alternative through:
\begin{itemize}
    \item \textbf{Named‑entity recognition (NER).} Deploy standard NER tools \cite{ehrmann2023named} to identify private information such as dates and locations.
    \item \textbf{Pre-constructed attribute database.} Maintain a user-curated mapping from private categories to keywords (e.g. {“job”: [“cybersecurity engineer”, “attorney”,...]}). Ontologies like DPV‑PD offer structured category definitions and synonym support for this purpose \cite{pandit2024data}. The dictionary can be used to match private attributes in the keyword, and the alternatives can be chosen from keywords of each category.
    \item \textbf{Local classifier.} Train a token-level or span-level model on annotated datasets or token-annotation lexicons to predict whether individual tokens or phrases are sensitive and their privacy categories. 
    \item \textbf{Mask-and-fill via online LLMs.} Replace sensitive spans in the prompt with \#MASK and request an online LLM to propose plausible substitutes. The validated outputs can be added to the alternative set for $A_{q_i}$.
\end{itemize}

To ensure the distinguishability between private and candidate attributes, we can enforce a recursive $(c,l)$-diversity constraint on the constructed candidate set to prevent. For a candidate set $S$ containing the true attribute and alternatives, let $\Pr(a|T)$ be the inferred probability of candidate attribute $a\in S$ given the generic instruction $T$. After sorting the probability distribution $p_1\geq p_2\geq \cdots \geq p_{|S|}$, we enforce:
\begin{equation}
    p_1 \leq c \sum_{i=l}^{|S|}p_i,
\end{equation}
which guarantees that attributes within the candidate set distributes more evenly.

\subsection{Collision of Generic Instruction}
In some cases, user prompts are highly specific or unique, resulting in very small clusters and limited queries for cost sharing. However, the cost for such extreme cases can be smoothed out by other common instructions on average. Both our theoretical analysis and empirical results show that the overall reduction is positively correlated with the entropy of the request distribution. In other words, the more skewed (long‑tailed) the distribution, the greater the expected amortization benefit, because a larger proportion of prompts have semantically equivalent generic instruction. 

To reduce per-prompt cost even in those extreme cases, we can adaptively adjust the batch window per generic instruction, waiting for sufficient queries before forwarding to the service provider. This dynamic batching introduces a natural latency–cost trade‑off: longer waiting times enable higher reduction ratios at the expense of increased delay.

Furthermore, the noise sampler can support users proactively by estimating the likely cluster size for a given generic instruction based on historical data. The system can warn users when the expected “collision” count is low. The user can then refine or generalize their instruction based on the sampler’s suggestions to increase cluster size.

\subsection{Query Deduplication}
\label{app:dedup}
Our framework’s cost analysis implicitly assumes that the service provider can deduplicate identical prompts at inference time and charge only once per unique evaluation. While such deduplication is technically feasible and implementable within a provider’s internal inference stack, current commercial LLM APIs might not generally perform this kind of deduplication by default. 

To address the gap between our cost model and real-world API billing, we propose incorporating an additional non-colluding deduplicator between the noise sampler and service provider:
\begin{itemize}
    \item \textbf{Encrypted Prompt Deduplication:} users encrypt their private attributes using the deduplicator’s public key. The deduplicator decrypts the private attributes to identify identical prompt for deduplication.
    \item \textbf{Service Provider Inference:} The deduplicated prompt set is sent to the service provider, which processes each unique prompt once and returns the corresponding outputs.
    \item \textbf{Response Expansion and Masking:} Upon receiving the responses, the deduplicator applies masking  and then expands the responses back to match each user’s original query set. These masked responses are then forwarded to the noise sampler while preserving privacy.
\end{itemize}
Under this design, the deduplicator’s view consists of a mixed set of original and noisy prompts that satisfies $(\mathcal{A}^n,\epsilon)$-indistinguishabilit, and by the post-processing property (see Proposition~\ref{prop:postprocess}), the service provider’s view achieves the same privacy guarantee. According to our empirical evaluation, incorporating the deduplicator incurs within 5\% additional computation cost under our experimental setting, where deduplication is parallelized across clusters. This allows the billing benefits of deduplication to be realized even under pricing models that do not natively support per-prompt deduplication, while still maintaining the privacy guarantees of our framework.

\end{document}